\documentclass[letterpaper, 10 pt, conference]{ieeeconf}  % Comment this line out
\IEEEoverridecommandlockouts                              % This command is only
\usepackage{amsmath} % assumes amsmath package installed
\usepackage{amssymb}  % assumes amsmath package installed
\usepackage{graphicx}
\usepackage[dvipsnames]{xcolor}

\usepackage{amsthm}
\usepackage{amssymb} %needed for \mathbb{}?
\usepackage{amsmath}
\usepackage{soul}

\usepackage{mathrsfs} % this is for mathscr

\usepackage{comment}

\usepackage{algorithm}
    \usepackage{algpseudocode}

\usepackage{cite}

\newtheorem{theorem}{Theorem}

\newtheorem{definition}{Definition}

\newtheorem{prop}{Proposition}

\newcommand{\blue}{\mathsf{b}}
\newcommand{\red}{\mathsf{r}}

\newcommand{\naturalSet}[1]{[#1]}
\newcommand{\actionsetblue}{\mathcal{A}^\blue{}}
\newcommand{\actionsetred}{\mathcal{A}^\red{}}

\newcommand{\extstrategy}{\gamma_{\text{ext}}}

\newcommand{\nodeset}{\mathcal{V}}
\newcommand{\edgeset}{\mathcal{E}}
\newcommand{\graph}{G}

\newcommand{\graphset}{\mathcal{G}}

\newcommand{\prior}{\rho}

\newcommand{\defaultblue}{\tilde{\gamma}^{\blue{}}_t}

\newcommand{\blueneighbors}[1]{{\mathcal{N}_{#1}}}

\newcommand{\redGraph}{\mathcal{G}^{\red}}
\newcommand{\redNodes}{{\mathcal{V}^{\red}}}
\newcommand{\redEdges}{\mathcal{E}^{\red}}
\newcommand{\redneighbors}[1]{{\mathcal{N}^{\red{}}_{#1}}}

\newcommand{\redtype}{\theta^\red}
\newcommand{\redtypeset}{{\Theta}^\red{}}

\newcommand{\goalset}{\mathcal{F}}
\newcommand{\bluetype}{{\theta^\blue}}
\newcommand{\bluetypeset}{{{\Theta}^\blue}}

\newcommand{\prob}{\mathbb{P}}

\newcommand{\restrictedpolicy}{\hat{\gamma}}
\newcommand{\restrictedactions}{\hat{\mathcal{A}}}
\newcommand{\restrictedpolicyset}{\hat{\Gamma}}
\newcommand{\iter}{\tau}
\newcommand{\bestresponse}{\mathbb{BR}}

\newcommand{\timebradds}{\tilde{T}}
\newcommand{\maxdistance}{{\overline{d}_\text{max}}}

\title{\LARGE 
\bf Deception and Counter Deception\\ in Adversarial Graph Traversal Game 
}

\author{Violetta Rostobaya, James Berneburg, and Daigo Shishika% idk what author order we're actually doing 
\thanks{$^1$Violetta Rostobaya, James Berneburg and Daigo Shishika are with College of Engineering and Computing at George Mason University, Emails: {\tt\footnotesize $\{$vrostoba,jbernebu,dshishik$\}$@gmu.edu}}%
}

\begin{document}

\maketitle
\thispagestyle{empty}
\pagestyle{empty}

\noindent
\begin{abstract}
We study deception in adversarial graph traversal, where a mobile agent seeks to reach a goal with minimum cost while an adversary alters edge costs to increase the total traversal cost. 
Unlike prior works that assume fixed observer–deceiver roles, we model this problem 
% as a sequential-move zero-sum stochastic shortest path game 
with two-sided incomplete information in which both players possess private information and update beliefs from observed actions. 
To solve the resulting indefinite-horizon game, we develop an adaptation of the Extensive-Form Double Oracle (XDO) algorithm. 
While the standard XDO algorithm is designed for finite games, the proposed adaptation ensures bounded computation despite endogenous game termination. 
We show that the proposed algorithm terminates in finite time and returns an $\epsilon$-Nash equilibrium. 
Finally, we use Value of Information to characterize the deceptive and counter-deceptive behaviors that emerge from equilibrium strategies.
\end{abstract}

\section{Introduction}
\noindent
In dynamic games with incomplete information, actions can be interpreted as signals that influence opponents’ beliefs about hidden information. 
When players update beliefs from observed game trajectories, actions affect both instantaneous cost and future responses. 
As a result, optimal strategies must balance short-term performance against long-term informational advantage.
This coupling between physical dynamics and belief dynamics distinguishes such problems from classical optimal control and planning.

A representative setting arises in adversarial graph traversal~\cite{berneburg2025multi}, where a mobile agents seek to reach a goal while minimizing traversal cost in an environment subject to adversarial modification. 
As illustrated in Fig.~\ref{fig:scenario}, 
% a team of robots navigates a graph whose
the edge costs may be altered, for example, by obstructing routes or degrading infrastructure. 
%In this setting, both players’ actions serve a dual role: they shape the physical evolution of the system while simultaneously revealing information about private objectives or capabilities. 
However, there may be uncertainties on the adversary's ability to attack the edges.
In such settings, actions affect not only the physical environment but also the information available to the opponent: adversarial actions may reveal capabilities, while the robots’ responses may reveal their final destination. Effective game play therefore requires planning under uncertainty.

Existing approaches do not fully capture this dual role of actions as both physical moves and sources of information. 
Classical graph-based planning methods, including shortest-path and stochastic optimization formulations, typically treat environmental changes as exogenous and do not account for strategic adversaries~\cite{ding2008finding,cai1997time,yuan2019constrained,pereira2013risk,primatesta2019risk,feyzabadi2014risk}. 
Goal recognition and deceptive path planning (DPP) frameworks study how an agent can manipulate an observer’s beliefs through its motion~\cite{dragan2015deceptive,masters2017deceptive,Ornik2018DeceptionIO,karabag2021deception}, however these works often assume predefined observer's model instead of treating observer as a decision maker that is capable of choosing their own belief model.
Game-theoretic models of navigation and attack-defense scenarios on graphs address adversarial decision making, but often assume complete information~\cite{berneburg2025multi,patek1999stochastic} or impose structural restrictions such as acyclic environments or limited feedback~\cite{durkota2015optimal,nguyen2017multi}. 
Game-theoretic DPP formulations consider a mobile deceiver facing a decision-making observer in graph environments~\cite{rostobaya2023deception, rostobaya2025deceptive,guan2025strategic}. These formulations, however, assume fixed roles, with one player observing and the other deceiving. 
In realistic adversarial settings, both players may possess private information and must balance deceptive and counter-deceptive behavior.

\begin{figure}
    \centering
    \includegraphics[width=\linewidth]{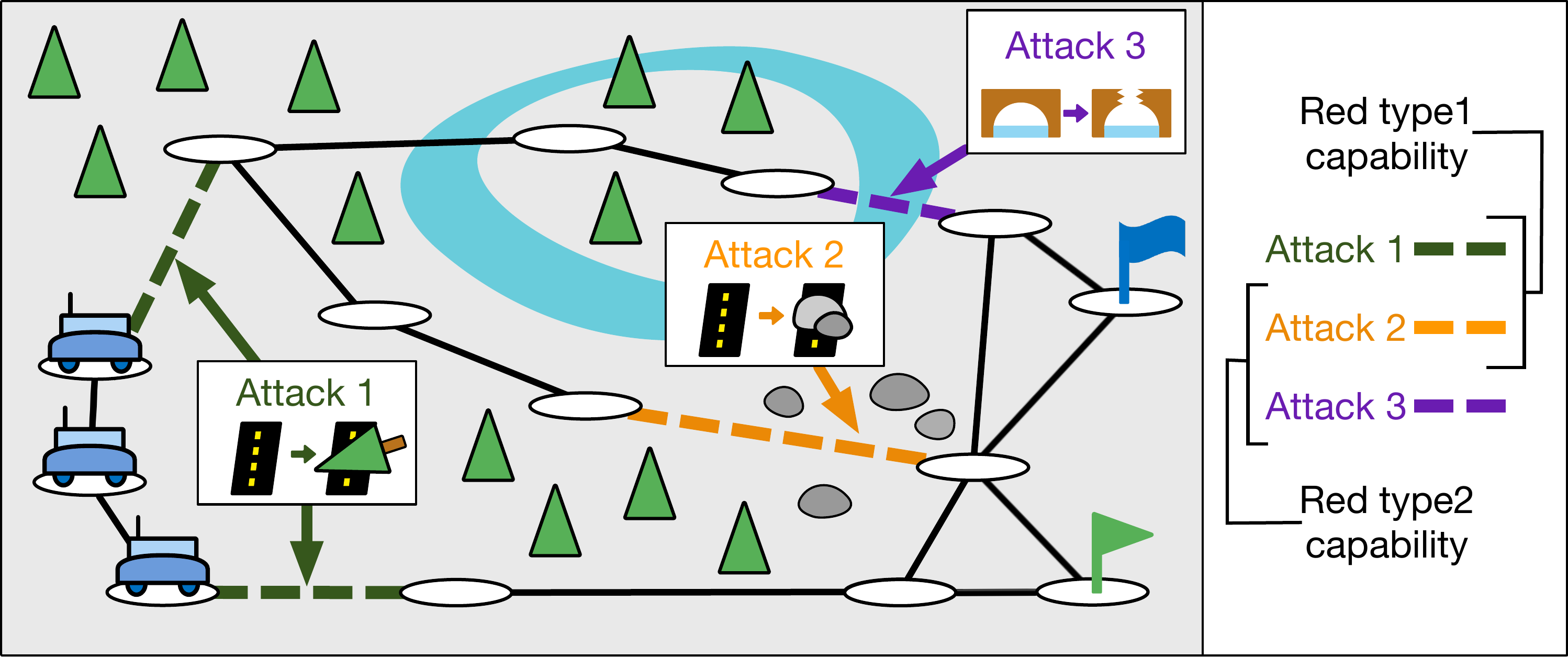}
    \caption{Illustration of incomplete-information adversarial graph traversal game. Blue robots must reach the assigned goal (one of two flags), while the Red player attacks the edges (e.g., by blocking roads) to maximize the traversal cost. Blue's goal and Red's  capability are private information, so each player must infer these based on observations of the opponents actions.}
    \label{fig:scenario}
\end{figure}

To address these limitations, we formulate an \emph{Adversarial Graph Traversal} (AGT) as a two-sided incomplete-information game. The game is modeled as a zero-sum stochastic shortest path game in which a mobile agent (Blue) seeks to reach a goal while an adversary (Red) dynamically modifies edge costs. 
Each player has private information, and both players update beliefs based on observed actions. 
Unlike prior DPP game formulations~\cite{rostobaya2023deception, rostobaya2025deceptive}, in our game both players influence the state evolution, and both must reason about how their actions affect the opponent’s beliefs. 
This setting requires both players to balance deception and counter-deception.

Solving incomplete-information games presents a computational challenge. 
Standard solution methods can not handle combination of two-sided incomplete information and an indefinite time horizon. 
Partially observable stochastic games (POSGs) provide a general modeling framework but are computationally intractable in large or long-horizon settings~\cite{horak2023solving,hansen2004dynamic}. 
Extensive-Form Double Oracle (XDO) methods can handle two-sided incomplete information, but require a finite extensive-form representation, which is not directly applicable to stochastic shortest path games with endogenous termination~\cite{mcaleer2021xdo, Boansk2014AnED}.
To overcome this challenge, we develop a modification of the XDO algorithm from~\cite{mcaleer2021xdo} tailored to indefinite-horizon stochastic shortest path games. 
The key idea is to introduce carefully designed default strategies that ensure finite termination and enable the construction of a finite restricted game at each iteration. 
This also allows best responses to be computed without propagating the decision trees indefinitely, while preserving the structure of the original problem. 
The resulting modified algorithm is shown to terminate in finite time and compute an $\epsilon$-Nash equilibrium.

The contributions of this paper are: (i) Formulation of a new adversarial graph traversal problem as a two-sided incomplete-information game, in which both players control the game and possess private information;
% state and update beliefs about each other types from observed histories;
(ii) Extension of the XDO algorithm for our AGT game with indefinite horizon by introducing default strategies that ensure finiteness of the restricted games and bounded computation of best-responses;
(iii) Proof on theoretical guarantees of finding $\epsilon$-NE; and 
(iv) Characterization of emergent deceptive and counter-deceptive behaviors through the analysis on game values and Value of Information.

\section{Problem Formulation}
\noindent
We consider a two-player stochastic game between a Blue player and a Red player. 
The Blue player controls a team of embodied agents navigating a graph environment and seeks to reach a goal configuration while minimizing traversal cost. 
The Red player can modify the graph’s edge costs over time in order to increase the cost incurred by the Blue player.
% Our formulation builds on the game introduced in~\cite{berneburg2025multi}, but differs in two key respects. 
% First, we adopt a sequential-move setting in which only one player acts at each time step, as opposed to the simultaneous-move formulation in~\cite{berneburg2025multi}. 
% Second, rather than assuming complete information, we consider an setting in which each player has private information, represented as a variable called \emph{type}, and each player must infer the opponent’s type through observations and belief updates.
We first describe the complete-information version of the game and then introduce asymmetric-information extensions.
% We consider a two-player stochastic game between a Blue player and a Red player. 
% The Blue player controls a team of embodied agents moving in a graph environment and seeks to reach a goal configuration with minimal traversal cost. 
% The Red player can alter the graph edge costs over time in order to increase the traversal cost incurred by the Blue player.
% The complete information version of the game was solved in~\cite{berneburg2025multi}. Our formulation differs from~\cite{berneburg2025multi}  primarily in two ways: i) ~\cite{berneburg2025multi} considered a simultaneous move setting, meanwhile  we consider a sequential move setting, where only one player acts each time step, and ii)~\cite{berneburg2025multi} considers complete knowledge of game state for both players, that is intentions and capabilities of all players is common knowledge, meanwhile in our game we consider an setting, where each player has private information that determines their goal or capability, and players need to infer their opponent's private information using observations and beliefs.
%We first describe the complete-information version of the game and then introduce two asymmetric-information extensions.

\paragraph*{Notation}
 We use $[n]\triangleq\{1,2,\dots,n\}$ for $n\in \mathbb{Z}_{>0}$. For a finite set $A$, let $\Delta(\mathcal{A}) \subset [0,1]^{|\mathcal{A}|}$ denote the probability simplex over $\mathcal{A}$.
We refer to a player using superscript $i \in \{\red,\blue\}$ and its opponent using index $-i\in \{\red,\blue\}\setminus\{i\}$, where $\blue$ is for the Blue player and $\red$ is for the Red player.

\subsection{Complete-information Game \cite{berneburg2025multi}}\label{subsection: complete info game} 
\noindent 
The Blue player controls an agent that moves on a weighted directed graph $\graph_t$, referred to as the terrain graph. 
Let $\graph_t=(\nodeset,\edgeset,w_t)$, where $\nodeset=[N]$ is a set of nodes for $N>0$, $\edgeset \subseteq \nodeset \times \nodeset$ is a set of edges, and $w_t:\edgeset\to \mathbb{R}_{>0}$ is a time-varying weight function that assigns a positive traversal cost to each edge. 
In particular, $w_t(p,p')>0$ denotes the cost of edge $(p,p')\in\edgeset$ at time $t$.
The terrain graph belongs to a known set of graphs $\graph_t\in \graphset\triangleq\{\graph^1,\graph^2,\dots,\graph^K\}$,
with corresponding weight functions $w_t\in \mathcal{W}\triangleq\{ w^1,w^2,\dots,w^K\}$, so that $\graph^k=(\nodeset,\edgeset,w^k)$ is a static graph for $k\in[K]$, where $K>0$ is the number of graphs.

%There are $M\in\mathbb{Z}_{>0}$ Blue team robots, and each robot $m$ at time $t$ occupies position $p_t^m\in \nodeset$. 
%We collect all robot positions into the vector $P_t=[p_t^1,p_t^2,\dots,p_t^M]^\top\in \nodeset^M$.
%The set of valid actions for the Blue player is $\actionsetblue=\blueneighbors{p}$, where $s$ is the game state defined below and $\blueneighbors{p}$ is the set of out-neighbors of node $p$ on the terrain graph. 
For a state $S_t=(p_t,\graph_t)$, where $p_t$ is the current Blue position, the set of valid actions for the Blue player is
$\actionsetblue(S_t)=\blueneighbors{p_t}$,
where $\blueneighbors{p_t}\triangleq \{p'\in\nodeset \mid (p_t,p')\in\edgeset\}$ is the set of out-neighbors of node $p_t$ on the terrain graph $G_t\in \mathcal{G}$.

The Red player's action is to select the graph at the next time step. 
Specifically, Red's available actions at time $t$ are specified by Red's \textit{action graph} $\redGraph=(\redNodes,\redEdges)$, where $\redNodes=[K]$ are the nodes of this graph and node $k$ corresponds to terrain graph $\graph^k\in\graphset$, while $\redEdges\subseteq\redNodes\times\redNodes$ is the set of edges. 
The action graph is directed, unweighted, and every node has a self-loop. 
The Red player's set of valid actions is $\actionsetred(p,\graph^k)=\redneighbors{k}$, 
%action is $
%a_t^\red\in \redneighbors{k}$,
where $\redneighbors{k}$ is the set of out-neighbors of node $k$ in $\redGraph$, and $\graph^k=\graph_t$. 
An example of action graph can be found in Fig.~\ref{fig: action graph}.
Since Red's available actions are determined by the current terrain graph, AGT models settings where multiple steps are required to change the environment or where changes on edge cost are irreversible. 
%Red's set of valid actions is $\actionsetred(P,\graph^k)=\redneighbors{k}$.

We define the state as $
S_t=(p_t,\graph_t)\in \mathcal{S}\triangleq \nodeset\times \graphset$.
We consider sequential actions, with the Blue player taking actions at even time steps, and the Red player taking actions at odd time steps.
The state dynamics for $t\geq0$ evolve according to
\begin{subequations}
\begin{align}
    p_{t+1}=
        a^\blue_t\in \actionsetblue(S_t),\; \graph_{t+1}=
    G_t \quad &\text{if}\;  t \; \text{is even},\\
        p_{t+1}=p_t,\;  G_{t+1}
    =\graph^{a^\red_t},  \; a^\red_t\in \actionsetred(S_t), \quad &\text{if} \;  t \; \text{is odd}.
    \end{align}
    % \begin{align}
    % \graph_{t+1}&=\begin{cases} 
    % G_t,\; &\text{if} \;  t \; \text{is even},\\
    %     \graph^{a^\red_t}, \; a^\red_t\in \actionsetred(S_t), &\text{if} \;  t\; \text{is odd}.
    % \end{cases}
    %  \end{align}
\end{subequations}

% The game terminates when Blue reaches a state in its goal set $ \goalset\subseteq \nodeset$.
% We assume there exists a state in $\goalset$ reachable from any state $s\in\mathcal{S}$. 
% % Therefore 
% The terminal time is
% $T=\min_{t\in \mathbb{Z}_{\geq 0}}\{t\mid p_t\in \goalset\}.$ 
% At each even time step, if Blue chooses action $a_t^\blue=p_{t+1}$, it incurs the stage cost
% \begin{equation}
%     C(s,a^\blue_t)=w_t(p_t,p_{t+1}), 
% \end{equation}
% and the cost is zero at odd time steps. 
The game terminates when Blue reaches a state in its goal set $\goalset\subseteq\nodeset$. 
We define the terminal time as $
T=\min_{t\in\mathbb{Z}_{\geq 0}}\{t\mid p_t\in\goalset\}$.
We assume that at least one state in $\goalset$ is reachable from any state $S_t\in\mathcal{S}$, so that termination is feasible from every initial state.
At each even time step, if Blue chooses action $a_t^\blue=p_{t+1}$, it incurs the stage cost
\begin{equation}
    C(S_t,a^\blue_t)=w_t(p_t,p_{t+1}), 
\end{equation}
and the cost is zero at odd time steps.

The objective of the Blue (resp. Red) player is to minimize (resp. maximize) the total expected cost of traversal from $S_0$ to $p\in \mathcal{F}$. The solution to the complete information game can be found using Shapley value iteration as explained in~\cite{berneburg2025multi}.
\subsection{Incomplete-information Game}\label{Subsection: incomplete-information game} 
\noindent
Now we consider the incomplete-information extension of the game. Each player has private information modeled using a variable called \emph{type}, and it forms a belief on the opponent's type. 
The Blue player possesses private information about its final goal configuration while the Red player possesses private information about its action graph.
%Both players have a variable called \emph{type} and 
Each player's type is denoted as $\theta^i$, $i \in \{\red,\blue\}$. Each type $\theta^i$ is randomly selected by Nature from a finite type set $\Theta^i$ according to commonly known prior distribution $\prior^i\in \Delta(\Theta^i)$. 
% The type specifies private information of the player $i$.

The  Blue player's type $\bluetype$ determines the Blue player's goal set $\goalset(\bluetype)\subseteq{\nodeset}$. 
Accordingly, the terminal time becomes
\begin{equation}\label{eq: termination}
T(\bluetype)=\min_{t\in \mathbb{Z}_{\geq 0}}\{t\mid p_t\in \goalset(\bluetype)\}.
\end{equation}
The Red player's type $\redtype$ determines Red's action graph
$\redGraph(\redtype)=(\redNodes,\redEdges(\redtype))$.
For each $\redtype\in\redtypeset$, we assume every node $k\in \redNodes$ has a self-loop.

Each player has perfect recall of state history and actions. 
The history of states at time $t$ is denoted $h_t$, with initial history $h_0=(p_0,\graph_0)$. The history evolves according to 
\begin{equation}
    h_{t+1}=(h_t,a^{i}_t),
\end{equation}
where $i=\blue$ if $t$ is even, and $i=\red$ if $t$ is odd.
The set of reachable histories at time $t$ is denoted by $\mathcal{H}_t$.
Each player's info set at time $t$ is denoted as $I^i_t=(\theta^i,h_t)\in \mathcal{I}^i_t= \Theta^i \times \mathcal{H}_t $, where $\mathcal{I}^i_t$ is player $i$'s set of info sets at time $t$.
% The Blue player's info set is $I^\blue_t=(\theta^\blue,h_t)\in \mathcal{I}^\blue_t= \Theta^\blue \times \mathcal{H}_t $, where $\mathcal{I}^\blue_t$ is Blue's set of info sets at time $t$. The Red player observes history $h_t$, its own type $\theta^\red$ and Blue's action $a^\blue_t$ at time $t$, i.e., Red's info set is $I^\red_t=(\theta^\red,h_t,a^\blue_t)\in \mathcal{I}^\red_t= \Theta^\red \times \mathcal{H}_t \times \actionsetblue(S_t) $.
We denote the set of admissible strategies for player $i$ as $\Gamma^i$, and let $\gamma^i$ be an element of this set. At time $t$, player $i\in \{ \red,\blue\}$ at info set $I^i_t$ selects an action from their action set $\mathcal{A}^i(S_t)$ according to their strategy $\gamma^i$. Specifically $\gamma^i$ is behavioral strategy with mapping $\gamma^i:\mathcal{I}^i\to\Delta(\mathcal{A}^i(S_t))$, where $\mathcal{I}^i:=\bigcup^\infty_{t=0} \mathcal{I}_t^i$.
Here strategy $\gamma^i(a^i_t|I^i_t=(\theta^i,h_t))$ denotes probability player $i$ type $\theta^i$ takes action~$a^i_t$ at history~$h_t$. 
% Here Blue's behavioral strategy $\bluestrategy^i(a_t^i\mid I^i_t=(\theta^i,h_t))$ denotes the probability player~$i$ type $\theta^i$ selects action $a_t^i$ given history $h_t$. 

Each player $i$ forms a belief about the the opponent's type, $\theta^{-i}$. 
Let $\mu^{i}(h_t)\in \Delta (\Theta^{-i}) $ denote the belief vector of player $i$, where $\mu^{i}(\theta^{-i}|h_t)=[\mu^i(h_t)]_{\theta^{-i}}\in [0,1]$ is the probability that player $-i$ is type $\theta^{-i}$  conditioned on the observed history $h_t$.
The initial belief is given by the commonly known prior $
\mu^{i}(h_0)=\prior^{-i}$.
For $t\ge1$, the belief $\mu^{i}(\theta^{-i} \mid h_t)$ is obtained recursively using Bayes' rule:
% To account for uncertainty about the Red player's type, the Blue player maintains a belief over $\redtypeset$. 
%  The Blue's belief is updated using Bayes' rule:
\begin{subequations} \label{eq: Bayes rule}
\begin{align} 
\mu^{i}(\theta^{-i} \mid h_{t+1})
&=
\frac{
\mu^{i}(\theta^{-i} \mid h_{t})\,
{\gamma^{-i}}(a_{t}^{-i} \mid \theta^{-i}, h_{t})
}{
\underset{\phi\in \Theta^{-i}}{\sum}
\mu^{i}(\phi \mid h_{t})\,
{\gamma^{-i}}(a_{t}^{-i} \mid \phi, h_{t})
}, 
\\
%\text{for}\; i&=\begin{cases}
%    \red{},\; \text{if} \;  t \; \text{is even},\\
%    \blue{}, \; \text{if} \;   t \; \text{is odd}, 
%\end{cases}\notag \\
\mu^{-i}(\theta^{i} \mid h_{t+1})&=\mu^{-i}(\theta^{i} \mid h_{t}), 
%\text{for}\;i&=\begin{cases}
%    \red{},\; \text{if} \;  t \; \text{is odd},\\
%    \blue{}, \; \text{if} \;  t \; \text{is even}, \notag
%\end{cases}
%\text{for}\; i&=\begin{cases}
%    \red{},\; \text{if} \;  t \; \text{is even},\\
%    \blue{}, \; \text{if} \;   t \; \text{is odd}, 
%\end{cases}\notag
\end{align}
\end{subequations}
where $i = \red$ if $t$ is even and $i = \blue$ if $t$ is odd. 
% where $h_{t+1}=(h_t,a^\blue_t,a^\red_t)$.
% Suppose the current history is $h_{t}$ and the Blue player's action $a_{t}^{\blue}$ is observed by the Red player. 
% Then the Red's updated belief is
% \begin{equation}\label{eq. Bayes rule, belief of Red}
% \mu^{\red}(\bluetype \mid h_t,a^\blue_t)
% =
% \frac{
% \mu^{\red}(\bluetype \mid h_{t})\,
% {\bluestrategy}_t^\blue(a_{t}^{\blue} \mid \bluetype,h_{t})
% }{
% \sum_{\phi\in\bluetypeset}
% \mu^{\red}(\phi \mid h_{t})\,
% {\bluestrategy}_t^\blue(a_{t}^{\blue} \mid \phi,h_{t})
% }.
% \end{equation} 

Both players aim to optimize for the expected total cost of the Blue player, where expectation is taken with resect to the player's strategies and the prior distributions $\prior=(\prior^{\red},\prior^\blue)$. The Blue player (resp. Red Player) aims to minimize (resp. maximize) the expected total cost of the Blue player:
\begin{equation}\label{eq:gameCost}
    J(\gamma^{\red},\gamma^{\blue};S_0,\prior)
    =
    \mathbb{E}
    \left[\sum_{t=0}^{T(\theta^\blue)} C(S_t,a_t^\blue)  \Bigm| \gamma^\red,\gamma^{\blue},\prior \right].
\end{equation}
We define the expected utilities of the Red player and the Blue player to be $U^\red(\gamma^{\red},\gamma^{\blue};S_0,\prior)=J(\gamma^{\red},\gamma^{\blue};S_0,\prior)$ and $U^\blue(\gamma^{\red},\gamma^{\blue};S_0,\prior)=-U^\red(\gamma^{\red},\gamma^{\blue};S_0,\prior)$, respectively. 
%We define the expected utility of the Red player to be $U^\red(\gamma^{\red},\gamma^{\blue};S_0,\prior)=J(\gamma^{\red},\gamma^{\blue};S_0,\prior)$, while the expected utility of the Blue player is $U^\blue(\gamma^{\red},\gamma^{\blue};S_0,\prior)=-U^\red(\gamma^{\red},\gamma^{\blue};S_0,\prior)$. 
Now, for initial state $S_0$ and distribution $\prior$, define player $i$'s best response to its opponent's strategy $\gamma^{-i}$ as 
\begin{equation}
\bestresponse^i(\gamma^{-i})\triangleq \text{arg}\max_{\gamma^{i}\in \Gamma^i}U^i(\gamma^{i},\gamma^{-i};S_0,\prior).
\end{equation}

Note that as a special case of this two-sided incomplete-information game, we can consider asymmetric-information cases in which only one player has private information. 
In this case, player $i$ has multiple possible types, i.e., $|\Theta^i|>1$, while the opponent has a single type, i.e., $|\Theta^{-i}|=1$.
Similarly, the complete information game setting described in Section~\ref{subsection: complete info game} is a special case
where both players have single types, i.e., $|\Theta^i|=1$, $\forall i\in\{ \red, \blue\}$.

\begin{definition}[$\epsilon$-Nash equilibrium]
    Strategy profile $\gamma^*=(\gamma^{\red*},\gamma^{\blue*})$ constitutes $\epsilon$-Nash Equilibrium if for each player $i\in \{ \red,\blue \}$ it satisfies:
    \begin{equation}
      U^i(\bestresponse^{i}(\gamma^{-i*}), \gamma^{-i*};S_0,\prior)\leq U^i(\gamma^*;S_0,\prior)+\epsilon,
        \end{equation}
        and the corresponding value of the game is defined as $V(S_0,\prior)\triangleq J(\gamma^*;S_0,\prior)$.
\end{definition}

 \subsection{Multi-agent scenario}
 \noindent
The formulation above extends to the case where the Blue player represents a decision maker for a team of agents. 
Following \cite{berneburg2025multi}, the multi-agent problem can be equivalently represented as a single Blue agent moving on a joint state terrain graph, whose nodes correspond to joint team configurations and whose edges correspond to feasible simultaneous actions of the team. Each edge weight on the joint graph is the total one-step cost of the corresponding joint action. 
%Under this representation, the formulation remains unchanged in structure: the Red type still determines the action graph, the Blue type still determines the goal set, and beliefs are updated from the observed history in the same way, except that the public state now contains the joint team configuration rather than a single robot position. 
%This representation is useful because it allows the multi-agent game to be analyzed within the same framework as the single-agent case, while still capturing coordinated behaviors among the Blue team agents.
Therefore, the multi-agent problem can be converted into a single-agent problem, so we analyze the single-agent problem for simplicity but without loss of generality. 

\section{Solution Method}
\noindent
The incomplete-information AGT game presents two computational challenges: both players may have private information, and the game has an indefinite terminal time. 
We use Extensive-Form Double Oracle (XDO) as the basis for our solution method because it can solve large extensive-form games by maintaining restricted game and iteratively expanding it by adding new best responses. 
This structure is appropriate for two-sided incomplete-information games, where directly enumerating the full strategy space is generally impractical. 
However, standard XDO assumes a finite extensive-form game, while the AGT game terminates endogenously when Blue reaches its goal. 
To address this issue, we adapt XDO by introducing default strategies that keep the restricted game and best-response computations finite.

We first review the XDO algorithm introduced in~\cite{mcaleer2021xdo}, and then adapt it to accommodate games with indefinite terminal time. 
The completeness of our algorithm is provided in Section~\ref{sec:theoretical}.

% We review the XDO algorithm introduced in \cite{mcaleer2021xdo}, and then adapt it to accommodate games with indefinite terminal time. 
% The completeness of our algorithm is provided in Section~\ref{sec:theoretical}.

\subsection{XDO Algorithm \cite{mcaleer2021xdo}}

%\todo{describe XDO algorithm here }
%\daigo{Don't we need to explain the concept of default strategies in this subsection?}

% XDO is an algorithm for solving \emph{finite} two player zero-sum games in extensive form which can handle two-sided It solves the game by iteratively expanding game tree, using a method such as Counterfactual Regret minimization (CFR), smaller restricted games defined by a limited set of pure strategies, and then adding new strategies to this set as needed. The algorithm is guaranteed to terminate for a finite game and yield $\epsilon$-NE strategies in a number of iterations which is linear in the number of information sets~\cite{mcaleer2021xdo}. More formally, let the set of pure strategies for player $i$ at iteration $\tau$ be $\restrictedpolicyset^i_\tau$, with some finite initial set $\restrictedpolicyset^i_0$. 
% Then for 
% The restricted extensive form game is then defined for the action set for each player $i$ 
% \begin{align}\label{eq:restrictedActionSet}
%     \restrictedactions^i_\iter(I^i) \triangleq \{ a^i \in A^i(S_t) | \exists \restrictedpolicy^i \in \restrictedpolicyset^i_\tau \text{ s. t. } \restrictedpolicy^i(a^i|I^i) = 1\}, 
% \end{align}
% where $S_t$ is the Markov state consistent with information set $I^i_t$. 
\noindent
XDO is an iterative algorithm for solving \emph{finite} two-player zero-sum games in extensive form that can handle two-sided incomplete information. 
%It iteratively expands the game tree. 
At each iteration, it solves smaller restricted games defined over a limited set of pure strategies using a method such as Counterfactual Regret Minimization (CFR), and then augments this set by adding new strategies as needed. For finite games, the algorithm is guaranteed to terminate and produce an $\epsilon$-Nash equilibrium in a number of iterations that is linear in the number of information sets~\cite{mcaleer2021xdo}.

More formally, let the set of pure strategies for player $i$ at iteration $\tau$ be $\restrictedpolicyset^i_\tau$, initialized with some finite set $\restrictedpolicyset^i_0$. The restricted extensive-form game is defined through the action set available to each player $i$ at an information set $I^i$ as
\begin{align}\label{eq:restrictedActionSet}
\restrictedactions^i_\iter(I^i) \triangleq \{ a^i \in \mathcal{A}^i(S_t) \mid \exists \restrictedpolicy^i \in \restrictedpolicyset^i_\tau \text{ s.t. } \restrictedpolicy^i(a^i | I^i) = 1 \}
\end{align}
where $S_t$ is the state consistent with the info state $I^i_t$.
%\james{We may need a more rigorous formulation for extensive form games for this. We at least need to indicate that the state in the above equation is info state and not Markov state. }
%Now, for initial state $S_0$ and distribution $\mathbf{p}=(\mathbf{p}^\red,\mathbf{p}^\blue)$, define player $i$'s best response to opponent's strategy $\gamma^{-i}$ as 
% \[
% \bestresponse^i(\gamma^{-i})\triangleq \text{arg}\max_{\gamma^{i}\in \Gamma^i}U^i(\gamma^{i},\gamma^{-i};S_0,\mu).
% \]
%Note that this is defined for the full game, not the restricted one. 

%With abuse of notation, 
We solve the restricted game using CFR, where the solution is computed over a subset of the full information structure. Specifically, let $\hat{\mathcal{I}}^i$ denote the set of info states included in the restricted game for player $i$, which captures the portion of the game currently being considered at a given iteration.
We can extend strategies for the restricted game to the strategies in the full game by using a pure \emph{default} strategy $\tilde{\gamma}^i\in \Gamma^i$ in any info state not present in $\hat{\mathcal{I}}^i$:
\begin{equation}\label{eq. full game strategy}
    \extstrategy^{i}(I^i)=\begin{cases}
        \hat{\gamma}^{i*}(I^i),\; \text{if} \; I^i \in \hat{\mathcal{I}}^i\\
        \tilde{\gamma}^i(I^i),\; \text{otherwise}.
    \end{cases}
\end{equation}
The default strategy $\tilde{\gamma}^i$ for each player can be chosen arbitrarily, allowing the strategy $\gamma^i_{\text{ext}}$ to be evaluated over the full game and enabling the computation of player $-i$’s best response. 
If the opponent's best response exploits player $i$’s strategy by more than $\epsilon$, then it is added to the strategy population for the next iteration, i.e., $\hat{\Gamma}^i_{\tau+1} = \hat{\Gamma}^i_{\tau} \cup \mathbb{BR}^i(\extstrategy^{-i})$ for each player $i$. 
This process continues until neither player can exploit the opponent’s strategy by more than $\epsilon$, at which point an $\epsilon$-Nash equilibrium is obtained and the algorithm terminates. 
The full XDO procedure is formally described in Algorithm~\ref{XDOalgorithm}, taking $\epsilon_1=\epsilon_2=\epsilon$, from~\cite{mcaleer2021xdo}.
\begin{algorithm} 
	\caption{Solve Unrestricted Game}\label{XDOalgorithm} 
	\begin{algorithmic}[1]
        \State Input: initial population $\restrictedpolicyset_0 = (\restrictedpolicyset_0^\blue,\restrictedpolicyset_0^\red)$, default strategies $\tilde{\gamma}^\blue$ and $\tilde{\gamma}^\red$, $\epsilon_1$ and $\epsilon_2$
        
        \State \textbf{repeat}
        \State \quad Define restricted game for $\restrictedpolicyset_\iter$ using~\eqref{eq:restrictedActionSet}
        \State \quad Get $\epsilon_1$-NE policy  $\restrictedpolicy^*$ of restricted game using CFR
        \State \quad Define $\extstrategy^i$ from~\eqref{eq. full game strategy}
        \State \quad Find $\bestresponse^i(\extstrategy^{-i})$ for each player $i$ 
        \State \quad \textbf{if} $U^i(\bestresponse^i(\extstrategy^{-i*}),\extstrategy^{-i*}) \leq U^i(\extstrategy^{*}) + \epsilon_2$ for both $i$ 
        \State \quad  \textbf{then} Terminate 
        \State \quad $\restrictedpolicyset_{\iter+1}^i = \restrictedpolicyset_\iter^i \cup \bestresponse^i(\extstrategy^{-i*})$ for each player $i$
	\end{algorithmic} 
\end{algorithm}

While this can be used to efficiently solve finite games, the indefinite terminal time of our game means that writing our game as a finite extensive form game is nontrivial.

\subsection{Default strategies}\label{subsection: default strategies}

\noindent
Although the choice of default strategies was not important in the original XDO algorithm of~\cite{mcaleer2021xdo}, arbitrary choices are no longer appropriate in our setting because the game does not have a pre-defined horizon. Specifically, our problem is a stochastic shortest-path game~\cite{patek1999stochastic} in which termination depends on Blue’s strategy, as given in~\eqref{eq: termination}.

This introduces three key challenges for applying Algorithm~1. First, there may exist a Blue player's pure strategy $\gamma^\blue \in \hat{\Gamma}^\blue_\tau$ under which the game never terminates, leading to an infinite restricted game tree that cannot be handled by CFR.
Secondly, even if the Blue's restricted game strategy $\hat{\gamma}^{\blue*}$ terminates the game, but the default strategy does not, then Red's best response would require solving an infinite decision tree.
Thirdly, even if Red’s strategy $\gamma^{\red}$ is fixed, it is defined over infinitely many information sets in $\mathcal{I}^\red$, so computing Blue’s best response would also require solving an infinite decision tree.

To address these issues, we introduce carefully designed default strategies for both players. The default strategy for Blue ensures that the unrestricted game terminates in finite time, thereby guaranteeing that the restricted game tree is also finite and that Red’s best response to $\hat{\gamma}^{\blue*}$ can be computed over a finite tree. 
Similarly, Red’s default strategy is constructed so that the Blue's best response also leads to finite-time termination, and it can be easily precomputed.

\paragraph*{Blue's default strategy $\defaultblue$}
We present a Blue player's default strategy. 
The key property of this strategy is that it guarantees termination of the overall unrestricted game in finite time; such a strategy is called \emph{proper}~\cite{patek1999stochastic}.

We first find maximum edge cost function:
\begin{equation*}
    \bar{w}(i,j)=\max_{w^k\in\mathcal{W}}w^k(i,j), \qquad \forall (i,j)\in \mathcal{E}.
\end{equation*}
We let $\bar{d}(i,j)$ be the shortest path distance on the terrain graph with the maximum edge cost function, i.e., on graph $\bar{G}=(\nodeset,\edgeset,\bar{w})$. 
We let $\tilde{\mathcal{A}}^\blue(S,\theta^\blue)$ be set of actions for Blue player type $\theta^\blue$ at state $S\in \mathcal{S}$ that Blue takes from $p\in \nodeset$ to stay on the shortest path from $p$ to closest node in $\mathcal{F}(\theta^\blue)$ denoted as $p_\text{goal}$, i.e.:
\begin{equation*}\label{eq:defaultstratoptblue}
 \tilde{\mathcal{A}}^\blue(S, \theta^\blue)= \underset{\substack{p_+\in \mathcal{A}^{\blue}(S),\; p_{\text{goal}}\in \mathcal{F}(\theta^{\blue})}} {\arg\!\min}\bar{d}(p,p_+)+\bar{d}(p_+,p_\text{goal}),
\end{equation*}
where $S=(p,G)\in \mathcal{S}$.
Let Blue's pure default strategy $\tilde{\gamma}^\blue$ deterministically select an action from $\tilde{\mathcal{A}}^{\blue}(S_t,\mathcal{\theta^{\blue}})$, i.e. $\exists a^\blue_t\in \tilde{\mathcal{A}}^{\blue}(S_t,\mathcal{\theta^{\blue}})$ such that $\tilde{\gamma}^{\blue}(a^\blue_t|I^\blue_t)\!= \!1$, where $S_t$ and $\theta^\blue$ are consistent with info set $I^\blue_t$.
%See that $\tilde{\gamma}^\blue$ is Markovian, i.e., depends only on current Blue's position.
See that $\tilde{\gamma}^\blue$ depends only on Blue's current position and is independent of Red's actions.

\paragraph*{Red's default strategy $\tilde{\gamma}^\red$} 
Red's pure default strategy $\tilde{\gamma}^\red$ is to simply stay at the current action graph node $k$, where $G_t=G^k$, an action that is available at any state $S=(p,G)\in\mathcal{S}$. Formally, $\tilde{\gamma}^\red(a^{\red}=k|I_t^\red)=1$, where $G_t=G^k$ and $G_t$ is consistent with Red's info set $I_t^\red$.
% \paragraph*{Blue's best response to ${\gamma}^\red$}
% Observe that the Red player's default strategy $\tilde{\gamma}^\red$ keeps the terrain graph $G_t$ stationary outside of the restricted game.
% As a result, it is straightforward to determine whether Blue’s best response exits the restricted game.
% In such cases, we assume that Blue proceeds along a shortest path on $G_t$ to the nearest node in $\mathcal{F}(\theta^\blue)$ for the remainder of the game. 
% Accordingly, we precompute Markovian policies on each $G^k \in \graphset$ that induce shortest paths to $\mathcal{F}(\theta^\blue)$. 
% Once the play exits the restricted region, Blue follows the appropriate Markovian policy for the rest of the game. 

% \begin{rem}[Structure outside the restricted game]
% If play reaches an information set outside the restricted game, then the Red player follows the default strategy $\tilde{\gamma}^\red$, and the terrain graph $G_t$ remains fixed thereafter.
% \end{rem}

\section{Theoretical Results}\label{sec:theoretical}
% some theoretical results would be nice but i don't think we have any yet \\
\noindent
% Results that help pruning the game tree?
Since the existing convergence proof for XDO applies only to finite games, we must rigorously justify why the use of Algorithm~1 with our choice of default strategies are valid for the AGT game which has indefinite-horizon problem.
% \daigo{Clarify what we're trying to do in Subsection A and B.}
We first establish that each iteration of Algorithm~1 does not produce infinite decision trees, and then we show that Algorithm~1 terminates in finite number of iterations.
Finally, we provide some analysis on the Value of Information that helps us characterize the strategies that emerge from our two-sided incomplete-information games.
\subsection{Bounded Computation per Iteration}
% \subsection{Avoiding Indefinite Game Tree Propagation}
\label{subsection: xdo for agt}
\noindent
Since the termination of the AGT game depends on Blue’s strategy (see~\eqref{eq: termination}), an arbitrary choice of default strategies or initial policies in $\hat{\Gamma}^i_0$ may result in infinite restricted-game or best-response decision trees. 
This is problematic because algorithms such as CFR and standard best-response computation are not directly applicable to infinite trees. We therefore formally show how our choice of default strategies, together with the initialization of the strategy population, avoids this issue.

\begin{prop}\label{prop: avoiding infinite decision trees}
    Consider the default strategies described in Section~\ref{subsection: default strategies} and let $\hat{\Gamma}^i_0 = \{\tilde{\gamma}^i\}$, i.e., the strategy population is initialized to include only the default strategies. Then each iteration step of Algorithm~\ref{XDOalgorithm} avoids solving an infinite-horizon decision problem. 
\end{prop}
% As a result, Algorithm~1 can be applied using finite restricted games despite the underlying indefinite-horizon formulation.

% We initialize the strategy population for each player $i$ to include only the default strategy, i.e., $\hat{\Gamma}^i_0 = \{\tilde{\gamma}^i\}$.

\begin{proof}
See appendix. 
\end{proof}

% \begin{proof}
% We must show that: (i)~restricted game defined in line~3 is finite; and (ii)~best-response computation for both players in line~6 is finite.

% For point (i), observe that Blue's default strategy defined in ... guarantees termination of the original (unrestricted) game in finite time since it is based on shortest paths. 
% Note that this is true against any Red's strategy. Consequently, the restricted game initialized by the specified default strategies is finite.

% For point (ii), first consider Blue's best response against Red's strategy. If ... remains within the restricted game, then the termination follows directly from the finiteness of the restricted game.  Whereas if ... then the terrain graph becomes fixed, and Blue's shortest path guarantees termination in finite time.  Therefore, every best response of Blue to ... can be found within a finite horizon.

% Finally, consider Red's best response...
% \qed
% \end{proof}

Proposition~\ref{prop: avoiding infinite decision trees} ensures that Algorithm~1 can be applied using finite restricted games despite the underlying indefinite-horizon formulation.

\subsection{Finite Time Convergence}
%\daigo{Motivate the Theorem.}
\noindent
Although we have shown that our adaptation of XDO allows us to run it on our game by ensuring that the computation at each iteration is finite, 
%restricted game and the computation of the best responses are always finite, 
this does not necessarily imply that the algorithm ever terminates. %, because the proof in~\cite{mcaleer2021xdo} relies on the finiteness of the unrestricted game. 
The following result therefore guarantees that this algorithm finds an equilibrium. 

\begin{theorem}\label{th:completeness}
    Given default strategies described in Section~\ref{subsection: default strategies}, the initialization of the strategy population in Section~\ref{subsection: xdo for agt}, and $\epsilon_1 <\epsilon_2/2$, Algorithm~\ref{XDOalgorithm} (XDO) terminates after finding an $\epsilon_2$-NE of the AGT game in finite number of iterations. 
    % applied to the AGT game terminates after a finite number of iterations with $\epsilon$-NE strategies for both players. 
\end{theorem}
% \james{Results to show that our algorithm does converge and/or terminate in finite time. }
{
\begin{proof}
See appendix. 
\end{proof}

This result shows that the proposed modification preserves a finite-time equilibrium computation guarantee for the original indefinite-horizon AGT game.

\subsection{Value of Information}
\noindent
To analyze deceptive and counter-deceptive behavior, we use the \emph{Value of Information} (VoI) associated with a player's private information. In words, VoI measures the relative utility loss caused by uncertainty about the opponent's type.

We use the subscripts CI, 1S$i$, and 2S to refer to the complete information, one-sided incomplete-information, and two-sided incomplete-information games, respectively. We treat the CI game as the game in which Nature randomly selects the players' types according to the prior $\prior$, and both types are publicly announced at the beginning of the game. Similarly, in the one-sided incomplete-information game 1S$i$, Nature randomly selects the players' types according to $\prior$, but publicly announces only the type of player $-i$, while the type of player $i$ remains private.

To compute the CI and 1S$i$ game values, let $\mathbf e_{\theta^i}\in\Delta(\Theta^i)$ denote the degenerate distribution on $\Theta^i$, which assigns probability~$1$ to type $\theta^i$. For the CI game, we first solve a collection of 2S games with degenerate priors $(\mathbf e_{\theta^\red}, \mathbf e_{\theta^\blue})$. We then take the expectation of the resulting values $V(S_0,(\mathbf e_{\theta^\red},\mathbf e_{\theta^\blue}))$ with respect to the original prior $\prior$, as shown in~\eqref{eq:V_CI}. The value of the 1S$i$ game is computed similarly in~\eqref{eq:V_1Si}. Note that the expectation with respect to $\prior^i$ in~\eqref{eq:V_1Si} is absorbed into the computation of $V(S_0,(\prior^i,\mathbf e_{\theta^{-i}}))$.

We consider three benchmark game values evaluated at $S_0$ and $\prior=(\prior^\red,\prior^\blue)$:
\begin{subequations}
\begin{align}
V_{\mathrm{CI}}
&\triangleq
\sum_{\theta^\red \in \Theta^\red}
\sum_{\theta^\blue \in \Theta^\blue}
\prior^\red(\theta^\red)\prior^\blue(\theta^\blue)\,
V(S_0,(\mathbf e_{\theta^\red},\mathbf e_{\theta^\blue})), \label{eq:V_CI}\\
V_{\mathrm{1S}i}
&\triangleq
\sum_{\theta^{-i} \in \Theta^{-i}}\prior^{-i}(\theta^{-i}) \,
V(S_0,(\prior^i,\mathbf e_{\theta^{-i}})), \label{eq:V_1Si}\\
V_{\mathrm{2S}}
&\triangleq
V(S_0,\prior). \label{eq:V_2S}
\end{align}
\end{subequations}
Here, \(V_{\mathrm{CI}}\) is the expected value of the complete-information game, \(V_{\mathrm{1S}i}\) is the expected value of the one-sided incomplete-information game in which only player \(i\) has private information, and \(V_{\mathrm{2S}}\) is the value of the two-sided incomplete-information game. 
We now define Value of Information.
\begin{definition} [Value of Information]
We let $\mathrm{VoI}_{\mathrm{1S}i}(i)$ and $\mathrm{VoI}_{\mathrm{2S}}(i)$ be the value of information for knowing player $i$'s type in $\textnormal{1S}i$ game and $\textnormal{2S}$ games respectively, that is:
\begin{subequations}
\begin{align*}
&\mathrm{VoI}_{\mathrm{1S}\red}(\red)
\triangleq
\frac{V_{\mathrm{1S}\red}-V_{\mathrm{CI}}}{V_{\mathrm{CI}}}, \; 
\mathrm{VoI}_{\mathrm{1S}\blue}(\blue)
\triangleq
\frac{V_{\mathrm{CI}}-V_{\mathrm{1S}\blue}}{V_{\mathrm{CI}}},\\
&\mathrm{VoI}_{\mathrm{2S}}(\red)
\triangleq
\frac{V_{\mathrm{2S}}-V_{\mathrm{1S}\blue}}{V_{\mathrm{1S}\blue}},\; 
\mathrm{VoI}_{\mathrm{2S}}(\blue)
\triangleq
\frac{V_{\mathrm{1S}\red}-V_{\mathrm{2S}}}{V_{\mathrm{1S}\red}}. \label{eq:VOI_2S}
\end{align*}
\end{subequations}
\end{definition}

Now we present the main result associated with VoI.

\begin{prop}\label{lem: VoI inequality 1}
The VoI for knowing $\theta^i$ in $\textnormal{1S}i$ game is greater than or equal to VoI for knowing $\theta^i$ in $\textnormal{2S}$ game, iff the opposite is true for the VoI for knowing~$\theta^{-i}$, i.e., 
    $\textnormal{VoI}_{\textnormal{1S}i}(i)\geq \textnormal{VoI}_{\mathrm{2S}}(i)$  iff $\textnormal{VoI}_{\textnormal{1S}-i}(-i)\leq \textnormal{VoI}_{\mathrm{2S}}(-i)$.
\end{prop}
\begin{proof}
    We show that the proposition's claim is true for~$i=\blue$. %and $-i=\red$. 
    The other case when~$i=\red$ %and $-i=\blue$ 
    can be proven in the same way. 
    % First observe that
    % \begin{align}
    %     0<V^{\textnormal{1S}\blue}\leq V^{\textnormal{CI}}\leq V^{\textnormal{1S}\red},\;
    %     0< V^{\textnormal{1S}\blue}\leq V^{\textnormal{2S}}\leq V^{\textnormal{1S}\red}.
    %         \end{align}
    Observe that the inequality $  \textnormal{VoI}_{\textnormal{1S}\blue}(\blue)\geq  \textnormal{VoI}_{\textnormal{2S}}(\blue)$ can be rewritten:
%\begin{equation*}
%\begin{split}
           $ \frac{V_{\textnormal{CI}}-V_{\textnormal{1S}\blue}}{V_{\textnormal{CI}}}\geq \frac{V_{\textnormal{1S}\red}-V_{\textnormal{2S}}}{V_{\textnormal{1S}\red}} \Leftrightarrow \frac{V_{\textnormal{1S}\red}}{V_\textnormal{CI}}\leq \frac{V_{\textnormal{2S}}}{V_{\textnormal{1S}\blue}} \Leftrightarrow
       \frac{V_{\textnormal{1S}\red}}{V_\textnormal{CI}}-1\leq \frac{V_{\textnormal{2S}}}{V_{\textnormal{1S}\blue}}-1 \Leftrightarrow\textnormal{VoI}_{\textnormal{1S}\red}(\red)\leq \textnormal{VoI}_{\textnormal{2S}}(\red)$.
%\end{split}
%\end{equation*}
\qed
\end{proof}
%Proposition~\ref{lem: VoI inequality 1} states that \emph{one and only one} player will improve the value of their private information, when facing uncertainty about opponent's type in the two-sided game, compared to the one-sided game.
Proposition~\ref{lem: VoI inequality 1} states that, relative to the one-sided game, \emph{one and only one} player can improve the value of their private information in the two-sided game, where each player is uncertain about the opponent's type.
}

\section{Numerical Illustration}
\noindent
This section provides examples that illustrate the deceptive and counter-deceptive behaviors that emerge as the solutions to the incomplete-information AGT game.
Notably, in our two-sided incomplete-information game, each player must simultaneously take advantage of its own private information while also mitigating the effect of uncertainty due to the opponent's private information.
To analyze the behaviors under those two possibly competing objectives, we consider cases where one player has only a single type, creating one-sided incomplete-information games.
We call these \emph{Blue's} game and \emph{Red's} game for conciseness.
These simplified settings allow us to extract purely deceptive and purely counter-deceptive behaviors against particular type.
Throughout the section, we consider the terrain graph and Red's action graphs depicted in Fig.~\ref{fig: action graph}.
\begin{figure}[h!]
    \centering
    \includegraphics[width=\linewidth]{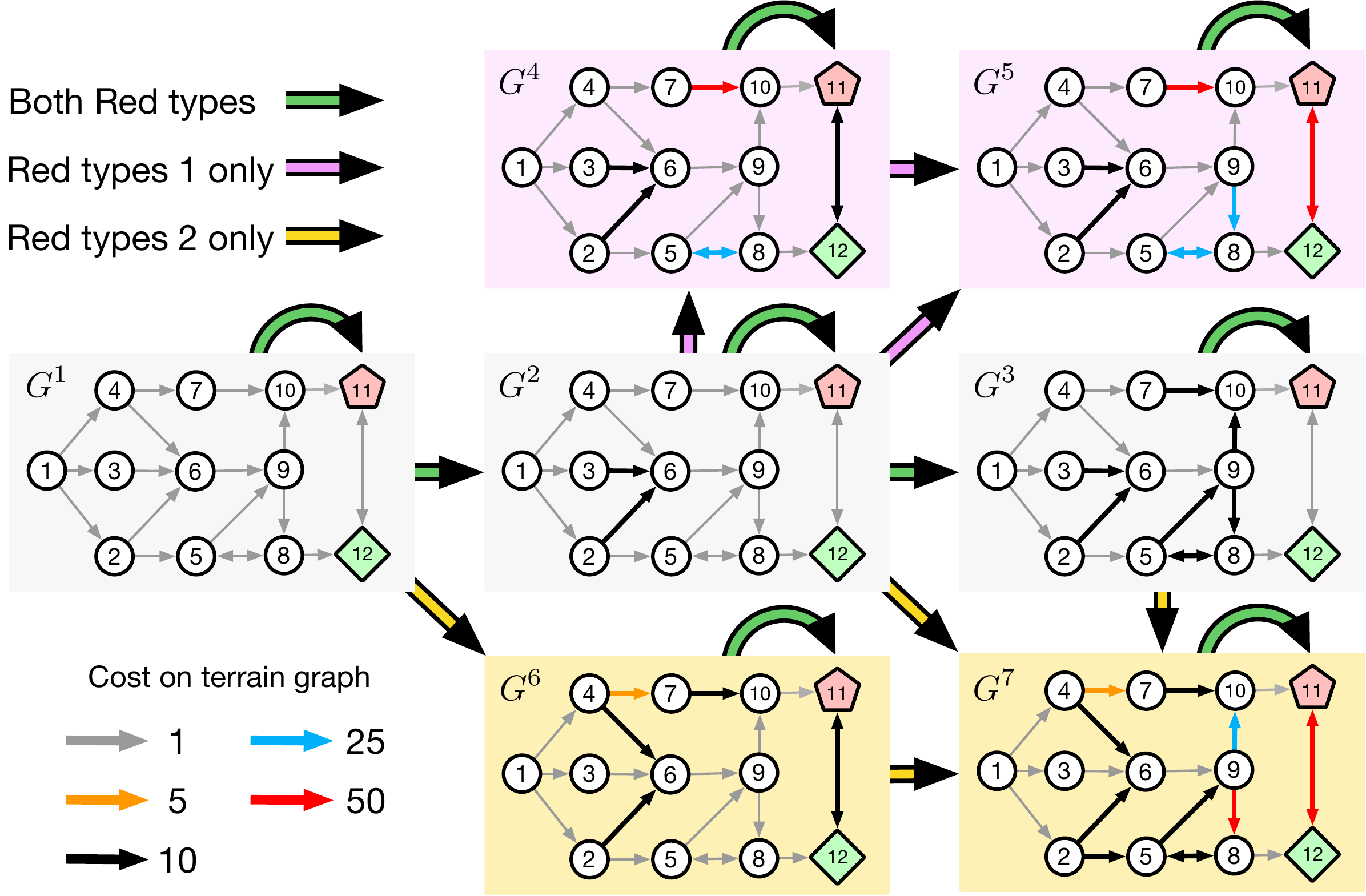}\caption{Action graph for Red player types~$1$ and~$2$. Both types can traverse the green edges, while only type~$1$ can traverse the pink edges and only type~$2$ can traverse the yellow edges. Since the game starts at~$G^1$, nodes $G^6$ and $G^7$ are unreachable for type~$1$, whereas nodes $G^4$ and $G^5$ are unreachable for type~$2$.}\label{fig: action graph}   \vspace{-5 pt}
\end{figure}
\begin{figure}[h!]
    \centering
    \includegraphics[width=\linewidth]{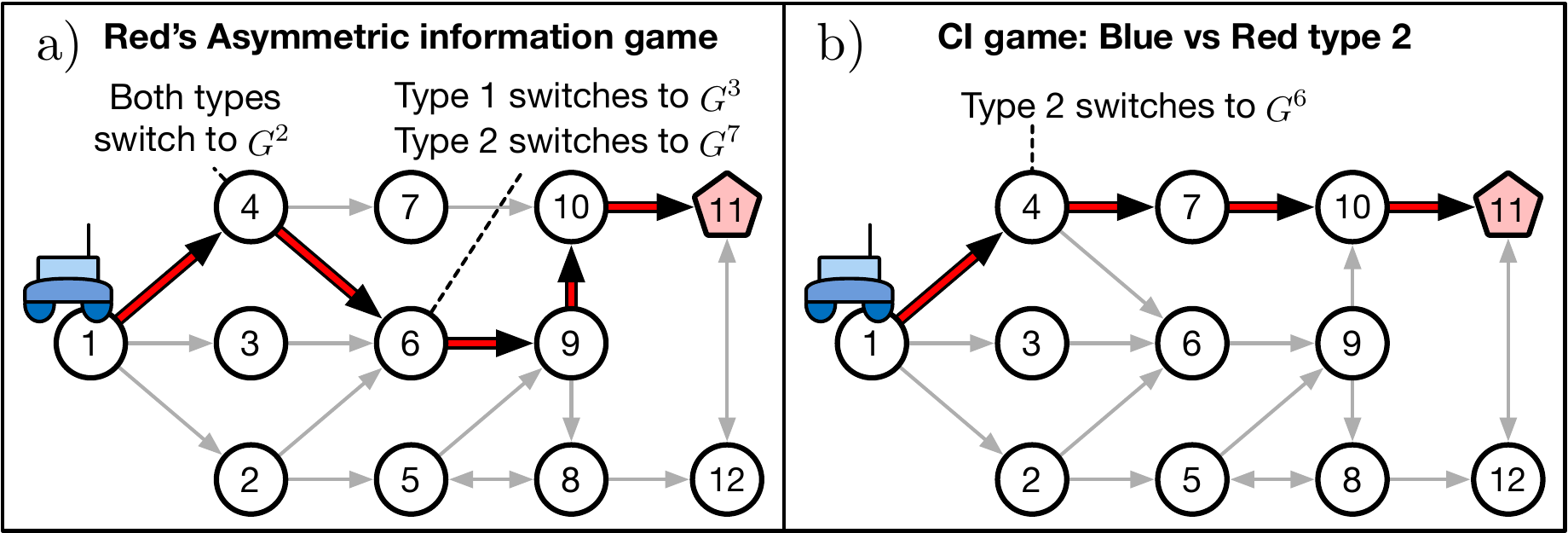}\caption{ a) Equilibrium behavior of the players in a game where only the Red player has a private information, payoff is~$17$, $\text{VoI}_{\text{1S}\red}(\red)=0.173$, Blue's belief is shown on Fig~\ref{Fig. 2 sided}c;  b) equilibrium behavior of the players in a complete information game, where Blue knows Red is type~$2$. }\label{Fig. reds game}  
    %\vspace{-15 pt}
\end{figure}
\noindent

% In this section we present deceptive behavior in three different cases: a) only Red player has private information, b) only Blue player has private information, and c) both players have private information. In case when the player does not have private information, its type set is singleton. Figures~\ref{Fig.action graph type 1} and~\ref{Fig.action graph type 2} show the set up for different types of Red player. 
In all three cases the game starts at $S_0=(p_0=1, G_0=G^1)$. The Blue player type~$1$ needs to reach node~$11$ on the terrain graph (shown as pentagon node), meanwhile type~$2$ needs to reach node~$12$ (shown as rhombus node), i.e. $\mathcal{F}(1)=\{11\}$ and $\mathcal{F}(2)=\{12\}$.

% \note{We won't have a lot of space, so what do we want to show in examples?} 

% To quantify the effect of private information, we define the Value of Information (VoI) for knowing type of player~$i$ defined as follows:
% \begin{equation}
% \mathrm{VoI}(i; \mathbf{p})
% =
% \frac{
% \bar U^{-i}(S_0,\mathbf p)
% -
% U^{-i}(\gamma^*;S_0,\mathbf p)
% }{
% \bar U^{-i}(S_0,\mathbf p)
% },
% \label{eq:voi_1s}
% \end{equation}
% \begin{equation*}
% \bar U^{-i}(S_0,\mathbf p)
% :=
% \sum_{\theta^i \in \Theta^i}
% \mathbf p^i(\theta^i)\,
% \bar{U}^{-i}(\bar{\gamma}^{*};S_0,(\mathbf{e}_{\theta^i},\mathbf{p}^{-i})),
% \label{eq:ci_benchmark_1s}
% \end{equation*}
% where $\mathbf{e}_{\theta^i}$ denotes the degenerate distribution on $\Theta^i$ that assigns probability one to type $\theta^i$. The term $\bar{U}^{-i}(\bar{\gamma}^{*};S_0,(\mathbf{e}_{\theta^i},\mathbf{p}^{-i}))$ is player $-i$'s utility in a game where player $i$'s realized with unit probability, and therefore player $-i$ automatically knows opponent's type $\theta^i$.
% The term $\bar U^{-i}(S_0,\mathbf p^i)$ is the benchmark, obtained by averaging the uninformed player’s utility across the games in which player $i$'s realized type is known to player $-i$.
% Thus, $\mathrm{VoI}(i;\mathbf{p})$ measures the relative utility loss for player~$-i$ caused by uncertainty about player $i$'s type.

\subsection{Red's Game: Only Red has private information}\label{sec. red's game} 

\noindent
Figure~\ref{Fig. reds game}a shows an example in which only the Red player has private information (Red's game), with type distribution $\prior^\red=[0.8,0.2]^\top$, meanwhile the Blue player's prior is degenerate, i.e. $\prior^\blue=[1,0]^\top$.
Under the equilibrium strategy, Red always switches to graph $G^2$ when Blue reaches node~$4$. When Blue reaches node $6$, Red type~$1$ switches to $G^3$, which increases edge cost of~$(9,10)$ from~$1$ to $10$, whereas  Red type~$2$ switches to $G^7$, which increases edge cost of~$(9,10)$ from~$1$ to $25$. 

% Under the equilibrium strategy, red type 1 attacks when blue reaches node 3, whereas red type 2 mixes between the following two: (i) attack immediately; and (ii) wait until blue reaches node 3.
% Since red's attack immediately reveals its type in this example, blue best responds once the attack occurs. Otherwise, if red waits holds the attack until blue is at node 3, it mixes between the upper route and lower route. Depending on the true type revealed at this time step, blue's best response is either (i) proceed to the goal; or (ii) back track to avoid the edge with highest cost.

In the complete information (CI) game (if Blue knows Red's type), Blue will take the route shown in Fig.~\ref{Fig. reds game}a against Red type~$1$ and the one in Fig.~\ref{Fig. reds game}b against type~$2$. As indicated in Fig.~\ref{Fig. reds game}b, Red type~$2$ maximizes its payoff by switching to $G^6$ which has a higher cost over the edge $(4,7)$.
Note that switching to $G^2$ as in Red's game will reduce the cost in this CI~game.
These observations imply that Red is sacrificing its immediate reward, but maintains its information advantage by selecting $G^2$ under both types in the incomplete-information game.

Red's equilibrium behavior maintains ambiguity in its type, so the Blue ``hedges'' by selecting the path $(1,4,6,9,10,11)$ for the expected cost of $17=0.2\cdot 29 + 0.8\cdot 14$. 
This path avoids a large penalty from using $(1,4,7,10,11)$ against Red type~$1$, who can make $(7,10)$ to have the cost of 50, resulting in the expected cost of $45.8=0.2\cdot17+0.8\cdot53$.

It is also notable that Red's strategy maintained Blue's belief at the prior, rather than manipulating it toward the uniform belief. Prior work suggests that the uniform belief is optimal from an information-theoretic perspective~\cite{FuJie2025,karabag2021deception,probine2024decentralized}. In contrast, our result shows that when belief manipulation is only a means of achieving the underlying mission objective, rather than the objective itself, the equilibrium belief need not be the maximum-entropy belief.
\begin{figure}[t]
    \centering
    \includegraphics[width=\linewidth]{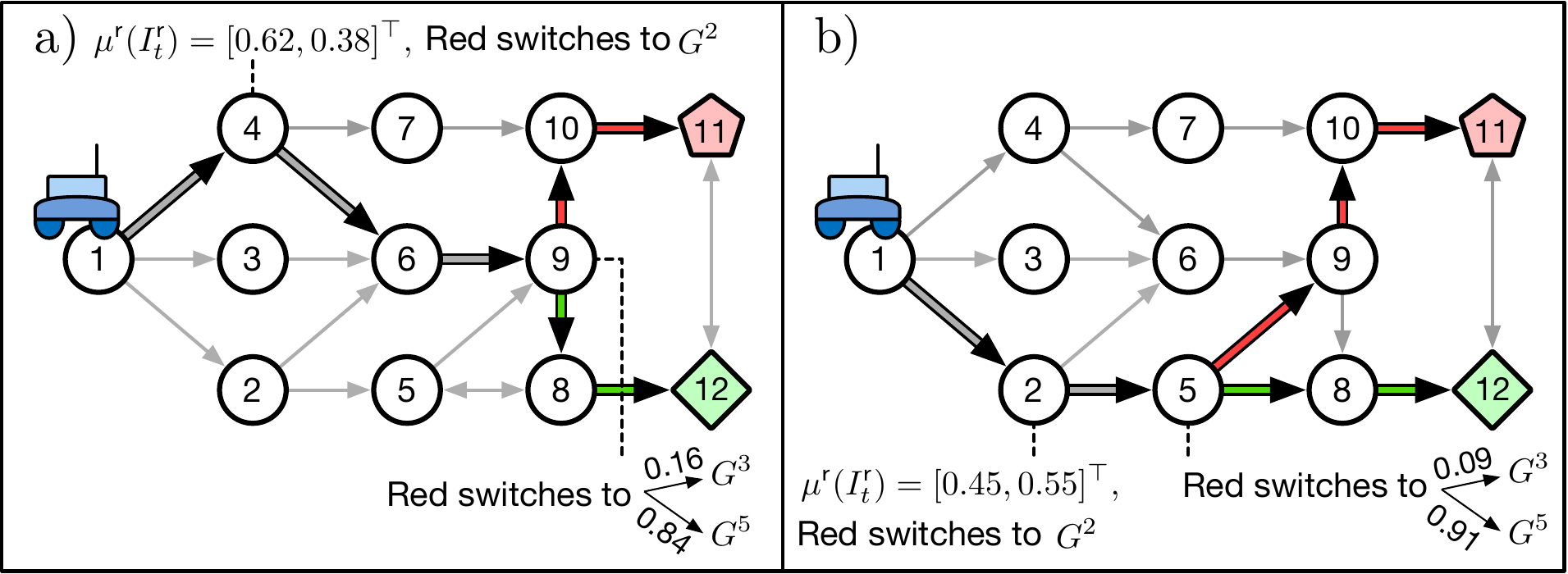}\caption{Equilibrium strategies in the Blue's game. Each Blue type mixes over two paths. Blue type~$1$ selects the gray–red path in a) and b) with probabilities $0.89$ and $0.11$, respectively. Blue type~$2$ selects the gray–green path in a) and b) with probabilities $0.81$ and $0.19$, respectively. Red's belief over Blue's type is noted at the nodes.}\label{Fig. blues game} 
    \vspace{-5pt}
    \end{figure}
% always choose the lower route to minimize the cost. This incentivized the red to attack immediately to max-min the traversal cost. 
% If red maintained ``ambiguous'' with probability 1, i.e., if type 2 never revealed early, the blue can exploit such strategy by always selecting the lower(?) route.  The corresponding value is XXX which is worse for the red.

% It is also notable that the red's strategy drove the belief to something that is non-uniform. 
% This in contrast to the DPP game, in which the belief manipulation drove it to uniform belief over type.
% Other literature also suggests such uniform belief to be optimal from information theoretic sense \todo{cite Jie Fu? and others}.
% Our result highlights that when belief manipulation is only a means to accomplish other mission (induce blue's uncertainty to maximize cost) and not the objective itself, then belief state other than highest-entropy one could be optimal.

\subsection{Blue's Game: Only Blue has Private Information}\label{sec: blue's game}
\noindent
Figure~\ref{Fig. blues game} shows \emph{Blue's game} in which only the Blue player has private information, corresponding to its goal nodes 11 and 12, with prior $\prior^\blue=[0.6,0.4]^\top$. Red is always type~$1$ from Fig.~\ref{fig: action graph}, which is known to Blue since $\prior^\red=[1,0]^\top$.

    % Equilibrium strategies in a game with asymmetric information, where the Blue player’s type determines its goal. 
% Blue type~$1$ selects the gray–red path in a) with probability $0.89$ and in b) with probability $0.11$. Blue type~$2$ selects the gray–green path in a) with probability $0.81$ and in b) with probability $0.19$.
% After observing Blue’s first move, Red’s belief over Blue’s type changes from $[0.6,0.4]^\top$ to $[0.62, 0.38]^\top$ in a) and to $[0.45, 0.55]^\top$ in b). 
% Blue’s type becomes fully revealed after node~$9$ in a) and node~$5$ in b).

Note that Fig.~\ref{Fig. blues game} shows the CI game~optimal path for Blue type~$1$ (gray-red path in~\ref{Fig. blues game}a) and~$2$~(gray-green path in~\ref{Fig. blues game}b).
This implies that each type mixes in a suboptimal path that mimics the behavior of the other type for the sake of manipulating Red's belief.
This deception continues until a point of divergence (nodes 9 or 5), after which the Red player's belief collapses to $[1,0]^\top$ or $[0,1]^\top$.
This uncertainty maintained by Blue's strategy forces Red to ``hedge'' by stochastically selecting either $G^3$ or $G^5$ before the type is revealed.
The expected traversal cost is $14.51$, which is lower than the value of CI~game of $19.6$, yielding $\text{VoI}_{\text{1S}\blue}(\blue) = 0.26$.

% Both Blue types are more likely to select the path in Fig.~\ref{Fig. blues game}a than in Fig.~\ref{Fig. blues game}b. In case (a), Blue type~$1$ follows its complete-information optimal path. In contrast, Blue type~$2$ initially mimics type~$1$ by following the same path up to node~$9$, after which it diverges and reveals its type.

% A less likely scenario occurs when Blue selects the path in Fig.~\ref{Fig. blues game}b. In this case, Blue type~$2$ follows its complete-information optimal path, while Blue type~$1$ mimics type~$2$ up to node~$5$ before deviating and revealing its type.

% In both cases, the Blue player maintains ambiguity over its type until a late stage of the trajectory. As a result, the Red player must hedge against this uncertainty and randomizes between $G^3$ and $G^5$ just prior to the point of type revelation.
Blue's behavior contrasts with the deception observed in the DPP game studied in~\cite{rostobaya2025deceptive}. In particular, in~\cite{rostobaya2025deceptive}, the mobile deceiver associated with the higher-prior type mixes between the its CI game path and a deceptive path, whereas the lower-prior type always follows a single path. 
In our example, by contrast, each type either follows its CI game path or deceives by following the CI game path of the other type up to the point at which the two paths diverge.

\subsection{Two-sided Incomplete-Information Game}
\noindent
Figure~\ref{Fig. 2 sided} illustrates a setting in which both players have private information.
The corresponding Red's action graphs are shown on Fig.~\ref{fig: action graph}.
The Red and Blue distributions are~$\prior^{\red}=[0.8,0.2]^\top$ and~$\prior^{\blue}=[0.6,0.4]^\top$.
% \begin{figure}[h!]
%     \centering
%     \includegraphics[width=\linewidth]{figures/2Sided.pdf}\caption{ Equilibrium strategies in a game with two-sided . 
%     Blue type~$1$ chooses gray-red paths in a) with probability~$0.56$. 
%     Blue type~$2$ always selects gray-green path. Blue's belief evolution after Blue chooses path in a) is shown on b).}\label{Fig. 2 sided}
% \end{figure}

Under equilibrium policies, the game can reduce to Red's game from Section~\ref{sec. red's game}, when Blue type~$1$ selects the upper path in Fig.~\ref{Fig. reds game}a with probability~$0.44$. 
Otherwise, Blue takes the lower path shown in Fig.~\ref{Fig. 2 sided}a. In this case Red initially preserves ambiguity by switching to $G^2$. 
Red type~1 either continues to conceal its type using $G^3$, or reveal itself by switching to $G^5$. Red type~$2$, by contrast, always uses $G^3$ when Blue is at node~$5$.
It reveals its type reveal later by switching to $G^7$. 
Consequently, if graph remains to be $G^3$ when Blue reaches node~8 or 9, it can immediately infer that Red is type~$1$, as shown in Fig.~\ref{Fig. 2 sided}d.

To see how the deception and counter-deception are balanced in 2S game strategies, we compare them to: (i)~CI strategies that do not have either; and (ii)~1S strategies that only contain deception or counter-deception, but not both.

\begin{figure}[t]
    \centering
    \includegraphics[width=\linewidth]{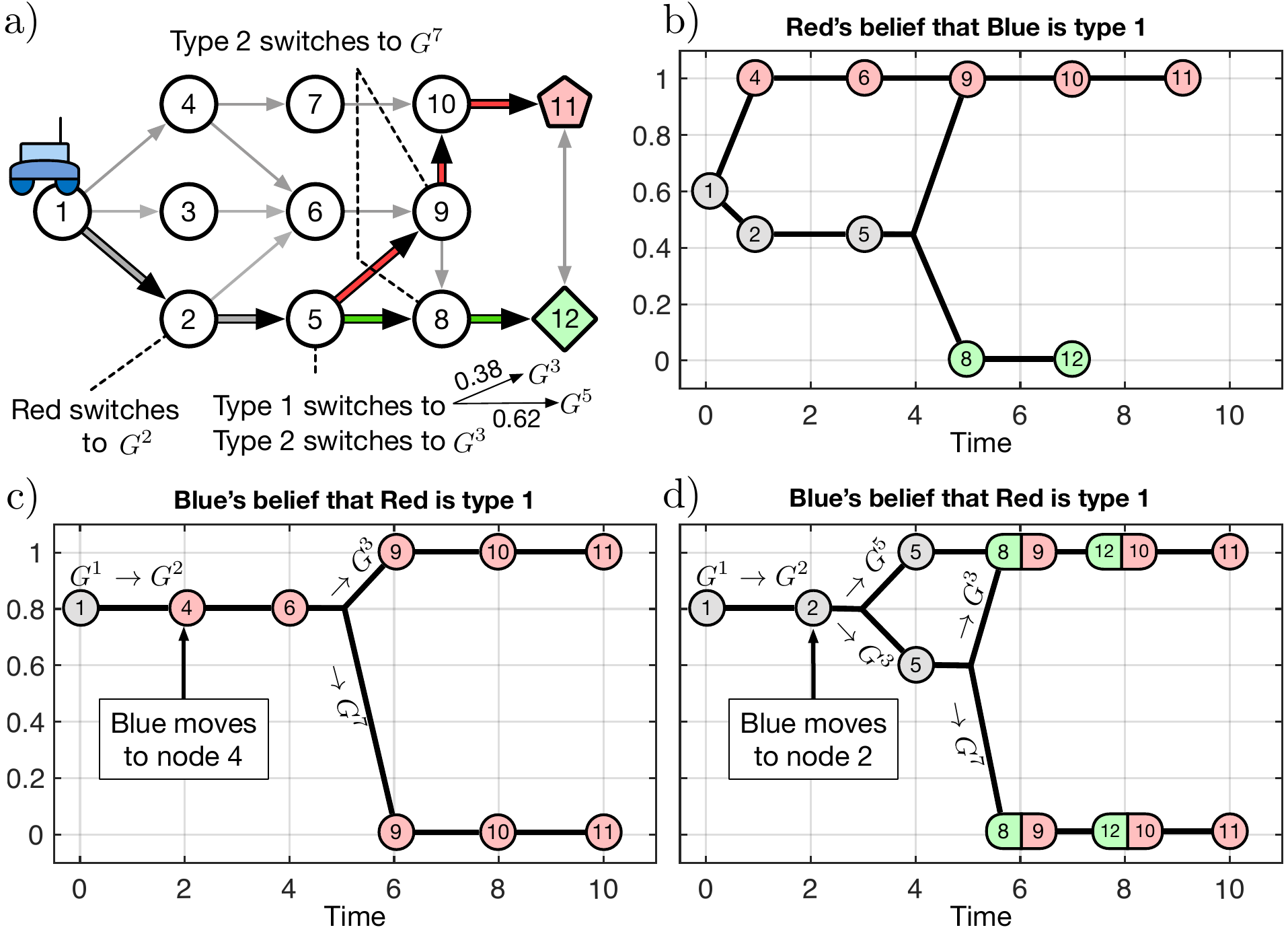}
    \caption{Equilibrium strategies in a game with two-sided incomplete information. 
Blue type~$1$ chooses gray-red paths in a) with probability~$0.56$, otherwise Blue type 1 chooses path in Fig.~\ref{Fig. reds game}a and game turns into Red's game.  Blue type~$2$ always selects gray-green path in a). Red's belief evolution is shown on b). Blue's belief evolution when $a^\blue_0=4$ and $a^\blue_0=2$ is shown on c) and d) respectively.}\label{Fig. 2 sided}
    \vspace{-10pt}
    \end{figure}

% \begin{figure}[h]
%     \centering
%     \includegraphics[width=\linewidth]{figures/Value_fig.pdf}
%     \caption{Comparison between different payoffs: a) $V_{\mathrm{1S}\blue}$, b)~$J_{\mathrm{1S}\red}(\gamma^{\red*}_\mathrm{2S},\bestresponse^\blue(\gamma^{\red*}_\mathrm{2S}))$, c)~$J_{\mathrm{1S}\red}(\gamma^{\red*}_\mathrm{2S},\bestresponse^\blue(\gamma^{\red*}_\mathrm{2S}))$.}\label{Fig. Values}
%     \vspace{-10pt}
%     \end{figure}    
\paragraph*{Blue's deceptiveness} For Blue's deceptiveness analysis, first consider the performance gap between deceptive 1S strategy, $\gamma^{\blue*}_{\mathrm{1S}\blue}$, and the CI strategy, $\gamma^{\blue*}_\mathrm{CI}$.  For the latter, we can evaluate the payoff when Red best responds: i.e., $J_{\mathrm{1S}\blue}(\bestresponse^\red(\gamma^{\blue*}_\mathrm{CI}),\gamma^{\blue*}_\mathrm{CI})$. The difference $\Delta J_{\mathrm{1S}\blue}(\mathrm{CI},\mathrm{1S}\blue)=J_{\mathrm{1S}\blue}(\bestresponse^\red(\gamma^{\blue*}_\mathrm{CI}),\gamma^{\blue*}_\mathrm{CI})-V_{\mathrm{1S}\blue}=4.06$ highlights the Blue's  performance gain from being optimally deceptive (with respect to the CI strategy). 
Now, if we evaluate Blue's 2S strategy in Blue's game, the performance gain compared to the CI strategy is: $\Delta J_{\mathrm{1S}\blue}(\mathrm{CI},\mathrm{2S})=J_{\mathrm{1S}\blue}(\bestresponse^\red(\gamma^{\blue*}_\mathrm{CI}),\gamma^{\blue*}_\mathrm{CI})-J_{\mathrm{1S}\blue}(\bestresponse^\red(\gamma^{\blue*}_\mathrm{2S}),\gamma^{\blue*}_\mathrm{2S})=0.32$. 
These values indicate that while 2S Blue strategy does not contain as much deceptiveness as its 1S strategy (since it is also being counter-deceptive against Red's uncertainty), it still performs better than the CI strategy which has no deception at all.

\paragraph*{Red's deceptiveness} A similar analysis for Red's 2S strategy shows $\Delta J_{\mathrm{1S}\red}(\mathrm{1S}\red,\mathrm{CI})=V_{\mathrm{1S}\red}-J_{\mathrm{1S}\red}(\gamma^{\red*}_\mathrm{CI},\bestresponse^\blue(\gamma^{\red*}_\mathrm{CI}))=1.54$ and $\Delta J_{\mathrm{1S}\red}(\mathrm{2S},\mathrm{CI})=J_{\mathrm{1S}\red}(\gamma^{\red*}_\mathrm{2S},\bestresponse^\blue(\gamma^{\red*}_\mathrm{2S}))-J_{\mathrm{1S}\red}(\gamma^{\red*}_\mathrm{CI},\bestresponse^\blue(\gamma^{\red*}_\mathrm{CI}))=-0.67$. This indicates that Red's 2S strategy performs worse than its CI strategy in the Red's game. 
This is in part due to the fact that 2S Red strategy, which assumes unknown opponent type, does not directly utilize the known Blue-type information given in Red's game. 
It also implies the limited deceptiveness of Red's 2S strategy, which is penalized in the Red's game.
Compared to the Blue's behavior discussed earlier, Red is performing more counter-deception in the 2S game.

\paragraph*{Comparative value of information}
% The equilibrium cost of traversal is $18.41$, and the value of information is
The equilibrium costs are $V_{\mathrm{CI}}=18.76$, $V_{\mathrm{1S}\red}=20.3$, $V_{\mathrm{1S}\blue}=14.7$ and $V_{\mathrm{2S}}=18.41$. The reduction in the game value from CI to 2S setting indicates that Blue benefited while Red suffered from the uncertainty in the 2S game.
When we look at the VoI, we obtain
$\text{VoI}_\mathrm{2S}(\blue)=0.09$, 
%($\bar U^{\red}(S_0,\mathbf p)=20.3$) 
$\text{VoI}_{\mathrm{2S}}(\red)=0.252$, 
%($\bar U^{\blue}(S_0,\mathbf p)=14.7$).
$\text{VoI}_{\mathrm{1S}\red}(\red)=0.08$, and $\text{VoI}_{\mathrm{1S}\blue}(\blue)=0.217$.
Note that value of Blue's private information decreased in 2S game in comparison with from 1S$\blue$ game, whereas the opposite happened for the Red player, which confirms statement of Proposition~\ref{lem: VoI inequality 1}.
%In the complete information game, where both players know each other's type, the expected payoff is 18.76, which indicates that the Blue player benefits more from the uncertainty about its own type and type of the Red player, despite its $\text{VoI}(\red;\mathbf{p})$ being high. 

%\daigo{Can we say something about the deceptiveness vs counter-deceptiveness in this two-sided scenario?}

%\subsection{Example with Blue player backtracking?}

\section{Conclusion}

\noindent
We formulated adversarial graph traversal as a sequential-move zero-sum stochastic shortest path game with two-sided incomplete information, in which both players update beliefs from observed actions and must balance deception with counter-deception. To solve this indefinite-horizon game, we developed a modification of the XDO algorithm based on specific default strategies, which ensures finite restricted games and finite best-response computation. We further established finite-time convergence of the proposed method to an $\epsilon$-NE. Numerical examples demonstrate how deceptive and counter-deceptive behaviors are balanced under equilibrium strategies and how they drive the value of information.

\appendix

\subsection{Proof of Proposition~\ref{prop: avoiding infinite decision trees}}
\noindent
We show by induction that at each iteration $\tau \geq 0$ of Algorithm~\ref{XDOalgorithm}, the restricted game constructed in line~3 is finite and the best-response computations in line~6 are carried out over finite trees.

For the initial case at $\tau=0$ where $\hat{\Gamma}_0^\blue=\{\tilde{\gamma}^\blue\}$, $\tilde{\gamma}^\blue$ follows a shortest path, with respect to $\bar d$, to the closest node in $\mathcal{F}(\theta^\blue)$. Hence it guarantees termination in finite time, regardless of Red's strategy, and is therefore proper. It follows that the initial restricted game tree is finite.

%\james{
Now assume that the restricted game tree is finite at iteration $\tau$ consider the best-response computations. 
%For each player $i$'s best response, if it only takes some action $a^i\in \hat{\mathcal{A}}_{\tau}^i(I^i)$, then it remains in the restricted game and so the decision tree is finite. %because the restricted game is finite. 
%First, for Blue's best response, if Blue only takes some action $a^\blue\in \hat{\mathcal{A}}_{\tau}^\blue(I^\blue)$, Blue remains in the restricted game, then the decision tree is finite because the restricted game is finite. 
Because it is finite, we only need to be concerned with actions that leave the restricted game, $a^i\notin \hat{\mathcal{A}}_{\tau}^i(I^i)$. 
If Blue leaves the restricted game, Red follows its default strategy $\tilde{\gamma}^\red$, after which the terrain graph becomes fixed and Blue can reach the closest node in $\mathcal{F}(\theta^\blue)$ in finite time by following a shortest path on the remaining fixed graph. 
Thus Blue's best-response tree is finite, and the resulting best-response strategy is proper.
%
%Second, for Red's best response, if Red only takes some action $a^\red\in \hat{\mathcal{A}}_{\tau}^\red(I^\red)$ and remains within the restricted game, finiteness again follows from the finiteness of the restricted game. 
If Red leaves the restricted game, Blue follows its default strategy $\tilde{\gamma}^\blue$, which is proper. Hence Red's best-response tree is finite as well.

Therefore, at any iteration $\tau$, if the restricted game tree is finite, then the best-response computations do not require solving an infinite-horizon decision problem.
Additionally, because the game terminates in finite time under the best response strategies, %the game tree is still finite when they are added, i.e. 
the restricted game tree is finite at iteration $\tau+1$ when they are added. %}
Therefore, by induction, at every iteration $\tau\geq 0$ of Algorithm~\ref{XDOalgorithm}$,$ the restricted game is finite and the best-response computations avoid solving an infinite-horizon decision problem. \qed

\vspace{-2.4 pt}
\subsection{Proof of Theorem~\ref{th:completeness}}
\noindent
By construction, if Algorithm~\ref{XDOalgorithm} terminates, then it provides $\epsilon_2$-NE strategies for both players, so we must show that it terminates after a finite number of iterations, i.e. that the condition in line 7 of Algorithm~\ref{XDOalgorithm} is satisfied for some $\tau^* < \infty$. 
%To this end, we seek to show that the condition in line 7 of Algorithm~\ref{XDOalgorithm} is satisfied after some finite number of iterations $\tau' < \infty$. 
To do this, we compare the restricted NE payoff of player $i$, i.e., $\hat{U}^{i*}\triangleq U^i(\extstrategy^{*};S_0, \prior)$ to the best response payoff of player $i$, i.e., $U^i_{\bestresponse} \triangleq U^i(\bestresponse^i(\extstrategy^{-i*}),\extstrategy^{-i*};S_0, \prior)$, where player $-i$ plays a restricted game NE strategy $\extstrategy^{-i*}$, to show that $U^i_{\bestresponse} - \hat{U}^{i*} \leq \epsilon_2 - \epsilon_1$ for $\tau$ large enough. Throughout the proof, player $i$ denotes the player that is best responding to the other's restricted NE policy. 

%The population $\restrictedpolicyset_\tau$ corresponds to restricted action sets $\restrictedactions^i_\iter(I^i)$ defined in~\eqref{eq:restrictedActionSet} and a set of info states $I^i$ which can be reached using those actions. 
%$I^i_t$

We begin with some concepts and notation. 
%Because of the $\epsilon$-NE guarantee in the restricted game, in order for $U^i_{\bestresponse} - \hat{U}^{i*} \geq \epsilon$, the best response $\gamma_\bestresponse^i \triangleq\bestresponse^i(\hat{\gamma}^{-i*})$ in line 6 must \emph{add an action} (perform an action $a^i \notin \restrictedactions^i_\tau(I^i_t)$) to an info set $I^i_t$. 
%Info state $I^i_t$ corresponds to some game history $h_t$ of length $t$. 
%
For $\gamma_\bestresponse^i\triangleq\bestresponse^i(\extstrategy^{-i})$, let $\timebradds$ be the first time that $\gamma_\bestresponse^i$ \emph{adds an action} (performs an action not in the restricted game). More formally, 
%\begin{align}
    $\timebradds \triangleq \min \{ t \geq 0 | \exists I^i_t \in \hat{\mathcal{I}}^i , a^i \notin \restrictedactions^i_\tau(I^i_t) 
    \text{ s. t. } \gamma_\bestresponse^i(a^i | I^i_t) > 0 \}$,
%\end{align}
%\cap \mathcal{I}_t^i
%\in \mathcal{A}^i(I^i_t)\setminus 
and $\timebradds = \infty$ if the minimum does not exist. $\timebradds< \infty$ is necessary for $U^i_{\bestresponse} - \hat{U}^{i*} > \epsilon_2 - \epsilon_1$, because of the $\epsilon_1$-NE guarantee in the restricted game. 
%$a^i \in \restrictedactions^i_\tau(I^i_t)$ for all $I^i_t$ and all $t < \timebradds$, so that no action is added before $\timebradds$. 
%the lowest $t$ for which there exists an information state $I^i_t$ such that $\bestresponse^i(\gamma^{-i})$ adds an action. 
%We will show that the difference $U^i_{\bestresponse} - \hat{U}^{i*} \leq \delta(t')$, where $\delta(t') : \mathbb{Z}_{> 0} \rightarrow \mathbb{R}_{>0}$ and $\lim_{t\rightarrow\infty} \delta(t')=0$. 
%We will show that the difference $U^i_{\bestresponse} - \hat{U}^{i*} \leq \delta(t')$, where $\delta : \mathbb{Z}_{> 0} \rightarrow \mathbb{R}_{>0}$ and $\lim_{t\rightarrow\infty} \delta(t)=0$. 
%We will show that as $\timebradds \rightarrow \infty$, the probability of the game ending before then approaches $1$. 
%

Due to the significance of adding an action, we split player $i$'s utility function across $\timebradds$. 
%
%We begin finding the function $\delta(t')$ by splitting player $i$'s policy and the payoff function across time.
%We write a \emph{combination} strategy, player $i$ plays policy $\gamma^i_1$ until time $\timebradds$ and plays policy $\gamma^i_2$ after time $\timebradds$, as 
%\begin{align}
%    \combinedstrategy^i_{\timebradds}(I^i_t;\gamma^i_1,\gamma^i_2)\triangleq \begin{cases}
%        \gamma^i_1(I^i_t), & t < \timebradds \\
%        \gamma^i_2(I^i_t), & t \geq \timebradds
%    \end{cases}. 
%\end{align}
%With abuse of notation, we rewrite the utility for player $i$
%, by rewriting~\eqref{eq:gameCost}, 
%as 
%\begin{align}
%    U^i(\gamma^i,\gamma^{-i};S_0,\mathbf{p}) = \mathbb{E}
%    \left[\sum_{t=0}^{\infty} C^i(t)  \right], 
%\end{align}
%where $C^\red(t) \triangleq C(S_t,a_t^\blue)$ and $C^\blue(t) \triangleq -C^\red(t)$, for $t \leq T(\theta^\blue)$ and $C^i(t) = 0$ otherwise. For convenience of notation, this extends the cost to sum over all time by assigning costs of zero after the game has ended. Now 
Abusing notation, the utility can be written as  
    $U^i(\gamma^i,\extstrategy^{-i*};S_0,\prior) = U^i_{<\timebradds}(\gamma^{i}) + U^i_{\geq \timebradds}(\gamma^{i})$, where 
%\end{align}
%where $\gamma^i = \combinedstrategy^i_{t'}(\gamma^i_1,\gamma^i_2)$ and
%\begin{align}
%\begin{align}
    $U^i_{<\timebradds}(\gamma^{i}) \triangleq \mathbb{E}
    [\sum_{t=0}^{\timebradds-1} C^i(t) | \gamma^{i},\extstrategy^{-i*} ]$ and %\notag \\
     $U^i_{\geq \timebradds}(\gamma^{i}) \triangleq \mathbb{E}
    [\sum_{t=\timebradds}^{\infty} C^i(t) | \gamma^{i},\extstrategy^{-i*} ]$,  
%\end{align}
and $C^\red(t) \triangleq C(S_t,a_t^\blue), C^\blue(t) \triangleq -C^\red(t)$ for $t \leq T(\theta^\blue)$, and $C^i(t) = 0$ otherwise. For convenient notation, this extends the summation to infinite time by assigning 
stage costs of zero after the game ends. 
%because the expectation operator is linear. 
%Notice that while $U^i_{<\timebradds}$ depends only on the strategy used before $\timebradds$, $U^i_{\geq \timebradds}$ depends on the strategy for all time because $\gamma^{i}_1$ determines the distribution of the info state at time $\timebradds$. 

%This split across time allows us to compare the payoff of player $i$'s best response strategy $\gamma_\bestresponse^i$ to the payoff of a \emph{restricted} best response strategy $\hat{\gamma}_\bestresponse^i\triangleq\combinedstrategy_{t'}^i(\gamma_\bestresponse^i,\tilde{\gamma}^i)$. 
With this notation defined, we provide an overview of the remaining proof. 
For a given $\timebradds$, we can upper bound $U^i_{\bestresponse} - \hat{U}^{i*}$ by first showing that the expected cost is bounded for the relevant blue policies. Since the stage costs at even time steps are lower bounded by a positive constant, the probability of the game \emph{not} ending before $\timebradds$ can be upper bounded, limiting the effect of adding actions on $U^i_{\bestresponse} - \hat{U}^{i*}$. %This can be used to show 
Then we show that $U^i_{\bestresponse} - \hat{U}^{i*} \leq \epsilon_2 - \epsilon_1$ for $\timebradds$, and hence $\tau$, large enough. 
%Because the stage costs are lower bounded by a positive constant,  the costs are upper bounded for the relevant blue policies

Now, rather than comparing the best response $\gamma_\bestresponse^i$ directly to the restricted NE policy $\extstrategy^{i*}$, we define 
%a more convenient restricted game strategy. 
the \emph{restricted} best response strategy to be 
\begin{align}
    \hat{\gamma}_\bestresponse^i %\triangleq \combinedstrategy_{\timebradds}^i(\gamma_\bestresponse^i,\tilde{\gamma}^i), \\
    \triangleq \begin{cases}
        \gamma_\bestresponse^i(I^i_t), & t < \timebradds \\
        \tilde{\gamma}^i(I^i_t), & t \geq \timebradds
    \end{cases}, 
\end{align}
which plays the default strategy after $\timebradds$, adding no actions. 
%. The latter strategy is a version of the best response strategy modified to be a restricted game strategy. 
%Because $t'$ is the first time $\gamma_\bestresponse^i$ plays an action not present in the restricted game, from that time forward $\hat{\gamma}_\bestresponse^i$ plays the default strategy so that it is a valid policy in the restricted game. 
%, which is always included in the restricted game because of the restricted game's initialization at $\tau=0$. 
%We write the payoff of this strategy against player $-i$'s NE strategy in the restricted game as 
Its payoff is $\hat{U}^i \triangleq U^i(\hat{\gamma}_\bestresponse^i,\extstrategy^{-i*};S_0,\prior)$. 
Due to properties of the restricted game's $\epsilon_1$-NE and the best response, we have $\hat{U}^i \leq \hat{U}^{i*}+\epsilon_1 \leq U^i_{\bestresponse}$. 
Therefore, if we can show that $U^i_{\bestresponse} - \hat{U}^i \leq \epsilon_2 - 2\epsilon_1$, then $U^i_{\bestresponse} - \hat{U}^{i*} \leq \epsilon_2 - \epsilon_1$ is implied. 
%
%
%This is convenient 
Because $\hat{\gamma}_\bestresponse^i$ and $\gamma_\bestresponse^i$ are identical before time $\timebradds$, we have $U^i_{\bestresponse} - \hat{U}^i = U^i_{\geq \timebradds}(\gamma_\bestresponse^i) - U^i_{\geq \timebradds}(\hat{\gamma}_\bestresponse^i)$. 
%We rewrite the remaining payoff 
For each $\gamma^i$, the remaining payoff is 
\begin{align}\label{eq:remainingexpectedcost}
    U^i_{\geq \timebradds}(\gamma^{i}) &= \sum_{I \in \mathcal{I}_{\timebradds}} \prob(I_{\timebradds})U^i_{\geq \timebradds}(\gamma^{i} | I_{\timebradds}), 
\end{align}
where
   $ U^i_{\geq \timebradds}(\gamma^{i} | I_{\timebradds}) \triangleq \mathbb{E} \left[\sum_{t=\timebradds}^{\infty} C^i(t) \Bigm|\gamma^i, \extstrategy^{-i*}, I_{\timebradds} \right]$. %\notag 
    %\sum_{I \in \mathcal{I}_{t'}} \prob(I_t) \mathbb{E} \left[\sum_{t=t'}^{\infty} C^i(t) \Bigm|\gamma^i, \gamma^{-i*}, I_{t'} \right].  

%where $\mathcal{I}_{\geq t'}$ is the set of all info states corresponding to histories of at least length $t'$. 

%We wish to upper bound this quantity, and 

We first find bounds for $U^i_{\geq \timebradds}(\gamma^{i} | I_{\timebradds})$. %to show that it is upper and lower bounded for all $I_{\timebradds}$ under the relevant strategies. 
%Let $\overline{v}(p)$ denote the value of the optimization in~\eqref{eq:defaultstratoptblue}, and $\overline{v}\triangleq \max_p \overline{v}(p)$. 
Let 
%\begin{align}
    \[\maxdistance \triangleq \max_{p \in \nodeset,\bluetype\in \bluetypeset} \min_{p_\text{goal}\in \goalset(\bluetype)} \overline{d}(p,p_\text{goal}).\] 
%\end{align}
By the default strategy's construction, we have $J(\gamma^{\red},\defaultblue;S_0,\prior) \leq \maxdistance$, for all $\gamma^{\red},S_0,\prior$. 
Because the default strategy is always in the restricted game, 
it can be shown that $|U^i_{\geq \timebradds}(\gamma^{i} | I_{\timebradds})| \in [0,\overline{v}]$, 
for $\gamma^{i} = \gamma_\bestresponse^i,\hat{\gamma}_\bestresponse^i$ and $i=\red,\blue$, where $\overline{v} \triangleq 2\maxdistance$. 
As a result, for $i=\red,\blue$, we have 
%\begin{align}
   $ U^i_{\geq \timebradds}(\gamma_\bestresponse^i | I_{\timebradds}) - U^i_{\geq \timebradds}(\hat{\gamma}_\bestresponse^i | I_{\timebradds}) \leq \overline{v}$. 
    %&\mathbb{E}\left[\sum_{t=t'}^{\infty} C^i(t) \Bigm| \gamma_\bestresponse^i, \gamma^{-i*}, I_{t'} \right] \notag\\
    %&- \mathbb{E}\left[\sum_{t=t'}^{\infty} C^i(t) \Bigm| \hat{\gamma}_\bestresponse^i, \gamma^{-i*}, I_{t'} \right] \leq 2\overline{v}. 
%\end{align}
%
Together with~\eqref{eq:remainingexpectedcost}, this implies 
%\begin{align}
    $U^i_{\bestresponse} - \hat{U}^i \leq \sum_{I_{\timebradds} \in \mathcal{I}_{\timebradds}} \prob(I_{\timebradds})\overline{v}$. 
%\end{align}

Next, we upper bound $\sum_{I_{\timebradds} \in \mathcal{I}_{\timebradds}} \prob(I_{\timebradds})$, %\prob(T(\theta^\blue) \geq \timebradds) = 
which is the probability that the game does not end before $\timebradds$. %, under the policy $\gamma_\bestresponse^i$. %, as a function of $\timebradds$. 
%
%To lower bound the cost of a trajectory of length $t'$, 
Define the minimum stage cost at even times as  
%\begin{align}
\[C_\text{min} \triangleq \min_{k\in \naturalSet{K}; (p,p') \in \edgeset} w^k(p,p'). \]
%\end{align}
%Therefore, for any game history $h_t$ with length $t$, the cost associated with it is lower bounded $U(h_t) > tC_\text{min}$. 
%Because the game remains in info states of the restricted game before time $\timebradds$, 
Because Blue either plays the restricted game NE strategy or a best response, we have 
\[\sum_{I_{\timebradds} \in \mathcal{I}_{\timebradds}} \prob(I_{\timebradds})C_\text{min}\timebradds/2 \leq U^i_{<\timebradds}(\gamma_\bestresponse^i)\leq \overline{v}.\] 
%This holds for both $i = \blue,\red$ because, for $i = \red$, Blue plays the restricted game equilibrium strategy and, for $i = \blue$, Blue plays a best response. 
Therefore, 
%\begin{align}
    %\sum_{I_{t'} \in \mathcal{I}_{t'}} \prob(I_{t'})C_\text{min}t' &\leq \overline{v}\notag \\
    $\sum_{I_{\timebradds} \in \mathcal{I}_{\timebradds}} \prob(I_{\timebradds}) \leq \frac{2\overline{v}}{C_\text{min}\timebradds}$ for $\timebradds>0$. 
%\end{align}
%Finally, then, $U^i_{\bestresponse} - \hat{U}^i \leq \delta(t) \triangleq \frac{2\overline{v}^2}{C_\text{min}t'}$. 
%Putting this all together, we have 
%Putting the bound on the difference %$U^i_{\geq \timebradds}(\gamma_\bestresponse^i | I_{\timebradds}) - U^i_{\geq \timebradds}(\hat{\gamma}_\bestresponse^i | I_{\timebradds})$ 
%together with this bound on the probability
Putting this bound together with the previous bound on the difference, we have 
%\begin{align}
    $U^i_{\bestresponse} - \hat{U}^{i*} -\epsilon_1\leq U^i_{\bestresponse} - \hat{U}^i \leq \frac{2\overline{v}^2}{C_\text{min}\timebradds}$.% \triangleq \delta(\timebradds)$.  
%\end{align}
%where $\Delta(t') \triangleq \frac{\overline{v}^2}{C_\text{min}t'}$. 
%

Therefore, because $\epsilon_2>2\epsilon_1$, there exists a time $\timebradds^*$ such that $\frac{2\overline{v}^2}{C_\text{min}\timebradds} \leq \epsilon_2-2\epsilon_1$ and so $U^i_{\bestresponse} - \hat{U}^{i*} \leq \epsilon_2$. %In words, if no actions are added before time $\timebradds^*$, then the payoff of the best response cannot be greater than the payoff of the restricted game by more than $\epsilon_2$. 
%Because there are a finite number of actions in each info state, there are only a finite number of info states 
%of length $\timebradds^*$, and at least one action is added to the restricted game at each iteration $\tau$, 
Finally, then, it can be shown that 
there exists a finite $\tau^*$ such that there are no actions to be added before time $\timebradds^*$, and so the difference $U^i_{\bestresponse} - \hat{U}^{i*} \leq \epsilon_2$ and the algorithm terminates. 
\qed

{\scriptsize

    \bibliographystyle{ieeetr}
    \bibliography{references}

    }

\end{document}